\documentclass[11pt]{article}
\usepackage[utf8]{inputenc}
\input{preamble.sty}
\title{Tensor Completion Made Practical}
\author{Allen Liu \thanks{Department of Electrical Engineering and Computer Science, Massachusetts Institute of Technology. Email: {\tt cliu568@mit.edu}.}\and Ankur Moitra\thanks{Department of Mathematics, Massachusetts Institute of Technology. Email: {\tt moitra@mit.edu}. This work was
supported in part by a Microsoft Trustworthy AI Grant, NSF CAREER Award CCF-1453261, NSF Large CCF1565235, a David and Lucile Packard Fellowship, an Alfred P. Sloan Fellowship and an ONR Young Investigator
Award.}}
\date{\today}

\newcommand{\wh}{\widehat}
\newcommand{\ov}{\overline}
\newcommand{\op}{\textsf{op}}
\newcommand{\spn}{\textsf{span}}
\renewcommand{\norm}[1]{\left\lVert#1\right\rVert}
\begin{document}

\maketitle

\begin{abstract}
    Tensor completion is a natural higher-order generalization of matrix completion where the goal is to recover a low-rank tensor from sparse observations of its entries. Existing algorithms are either heuristic without provable guarantees, based on solving large semidefinite programs which are impractical to run, or make strong assumptions such as requiring the factors to be nearly orthogonal. 

In this paper we introduce a new variant of alternating minimization, which in turn is inspired by understanding how the progress measures that guide convergence of alternating minimization in the matrix setting need to be adapted to the tensor setting. We show strong provable guarantees, including showing that our algorithm converges linearly to the true tensors even when the factors are highly correlated and can be implemented in nearly linear time. Moreover our algorithm is also highly practical and we show that we can complete third order tensors with a thousand dimensions from observing a tiny fraction of its entries. In contrast, and somewhat surprisingly, we show that the standard version of alternating minimization, without our new twist, can converge at a drastically slower rate in practice. 
\end{abstract}

\newpage

\section{Introduction}

In this paper we study the problem of recovering a low-rank tensor from sparse observations of its entries. In particular, suppose
$$
T = \sum_{i=1}^r \sigma_i \mbox{ } x_i \otimes y_i \otimes z_i
$$
Here $\{x_i\}_i$, $\{y_i\}_i$ and $\{z_i\}_i$ are called the {\em factors} in the low rank decomposition and the $\sigma_i$'s are scalars, which allows us to assume without loss of generality that the factors are unit vectors. Additionally $T$ is called a {\em third order} tensor because it is naturally represented as a three-dimensional array of numbers. Now suppose each entry of $T$ is revealed independently with some probability $p$. Our goal is to accurately estimate the missing entries with $p$ as small as possible. 

Tensors completion is a natural generalization of the classic matrix completion problem \cite{candes2009exact}. Similarly, it has a wide range of applications including in recommendation systems \cite{zhu2018fairness}, signal and image processing \cite{liu2012tensor, li2017low, ng2017adaptive, bengua2017efficient}, data analysis in engineering and the sciences \cite{tan2016short, kreimer2013tensor, wang2015rubik, xie2016accurate} and harmonic analysis \cite{trickett2013interpolation}. However, unlike matrix completion, there is currently a large divide between theory and practice. Algorithms with rigorous guarantees either rely on solving very large semidefinite programs \cite{barak2016noisy, potechin2017exact}, which is impractical, or else need to make strong assumptions that are not usually satisfied \cite{cai2019nonconvex}, such as assuming that the factors are nearly orthogonal, which is a substantial restriction on the model. In contrast, the most popular approach in practice is {\em alternating minimization} where we fix two out of three sets of factors and optimize over the other
\begin{equation}\label{eq:naivealtmin}
(\wh{z}_1, \dots , \wh{z}_r) = \arg\min_{z_1, \dots , z_r}\norm{\left(T - \sum_{i=1}^r \wh{x}_i \otimes \wh{y}_i \otimes z_i\right)\Bigg|_{S}}_2^2
\end{equation}
Here $S$ is the set of observations and we use $X \Big |_{S}$ to denote restricting a tensor $X$ to the set of entries in $S$. We then update our estimates, and optimize over a different set of factors, continuing in this fashion until convergence. A key feature of alternating minimization that makes it so appealing in practice is that it only needs to store $3r$ vectors along with the observations, and never explicitly writes down the entire tensor. Unfortunately, not much is rigorously known about alternating minimization for tensor completion, unlike for its matrix counterpart \cite{jain2013low, matrixaltmin}. 

In this paper we introduce a new variant of alternating minimization for which we can prove strong theoretical guarantees. Moreover we show that our algorithm is highly practical. We can complete third order tensors with a thousand dimensions from observing a tiny fraction of its entries. We observe experimentally that, in many natural settings, our algorithm takes an order of magnitude fewer iterations to converge than the standard version of alternating minimization.

\subsection{Prior Results}

 In matrix completion, the first algorithms were based on finding a completion that minimizes the nuclear norm \cite{candes2009exact}. This is in some sense the best convex relaxation to the rank \cite{chandrasekaran2012convex}, and originates from ideas in compressed sensing \cite{fazel2002matrix}. There is a generalization of the nuclear norm to the tensor setting. Thus a natural approach \cite{yuan2016tensor} to completing tensors is to solve the convex program
$$\min \|\wh{T}\|_* \mbox{ s.t. } \wh{T} \Big |_{S} = T \Big |_{S}$$
where $\|\cdot \|_*$ is the tensor nuclear norm. Unfortunately the tensor nuclear norm is hard to compute \cite{gurvits2003classical, hillar2013most}, so this approach does not lead to any algorithmic guarantees. 

Barak and Moitra \cite{barak2016noisy} used a semidefinite relaxation to the tensor nuclear norm and showed that an $n \times n \times n$ incoherent tensor of rank $r$ can be recovered approximately from roughly $r n^{3/2}$ observations. Moreover they gave evidence that this bound is tight by showing lower bounds against powerful families of semidefinite programming relaxations as well as relating the problem of completing approximately low-rank tensors from few observations to the problem of refuting a random constraint satisfaction problem with few clauses \cite{daniely2013more}. This is a significant difference from matrix completion in the sense that here there are believed to be fundamental computational vs. statistical tradeoffs whereby any efficient algorithm must use a number of observations that is larger by a polynomial factor than what is possible information-theoretically. Their results, however, left open two important questions:

\begin{aqquestion}
Are there algorithms that achieve exact completion, rather than merely getting most of the missing entries mostly correct?
\end{aqquestion} 

\begin{aqquestion}
Are there much faster algorithms that still have provable guarantees \--- ideally ones that can actually be implemented in practice? 
\end{aqquestion}

For the first question, Potechin and Steurer \cite{potechin2017exact} gave a refined analysis of the semidefinite programming approach through which they gave an exact completion algorithm when the factors in the low rank decomposition are orthogonal.  Jain and Oh \cite{jain2014provable} obtain similar guarantees via an alternating minimization-based approach. For matrices, orthogonality of the factors can be assumed without loss of generality. But for tensors, it is a substantial restriction. Indeed, one of the primary applications of tensor decomposition is to parameter learning where the fact that the decomposition is unique even when the factors can be highly correlated is essential \cite{anandkumar2014tensor}. Xia and Yuan \cite{xia2019polynomial} gave an algorithm based on optimization over the Grassmannian and they claimed it achieves exact completion in polynomial time under mild conditions. However no bound on the number of iterations was given, and it is only known that each step can be implemented in polynomial time. For the second question, Cai et al. \cite{cai2019nonconvex} gave an algorithm based on nonconvex optimization that runs in nearly linear time up to a polynomial in $r$ factor. (In the rest of the paper we will think of $r$ as constant or polylogarithmic and thus we will omit the phrase ``up to a polynomial in $r$ factor" when discussing the running time.) Moreover their algorithm achieves exact completion under the somewhat weaker condition that the factors are nearly orthogonal. Notably, their algorithm also works in the noisy setting where they showed it nearly achieves the minimax optimal prediction error for the missing entries. 

There are also a wide variety of heuristics. For example, there are many other relaxations for the tensor nuclear norm based on flattening the tensor into a matrix in different ways \cite{gandy2011tensor}. However we are not aware of any rigorous guarantees for such methods that do much better than completing each slice of the tensor as its own separate matrix completion problem. Such methods require a number of observations that is a polynomial factor larger than what is needed by other algorithms. 

\subsection{Our Results}

In this paper we introduce a new variant of alternating minimization that is not only highly practical but also allows us to resolve many of the outstanding theoretical problems in the area. Our algorithm is based on some of the key progress measures behind the theoretical analysis of alternating minimization in the matrix completion setting. In particular, Jain et al. \cite{jain2013low} and Hardt \cite{matrixaltmin} track the principal angle between the true subspace spanned by the columns of the unknown matrix and that of the estimate. They prove that this distance measure decreases geometrically. In Hardt's \cite{matrixaltmin} analysis this is based on relating the steps of alternating minimization to a noisy power method, where the noise comes from the fact that we only partially observe the matrix that we want to recover. 

We observe two simple facts. First, if we have bounds on the principal angles between the subspaces spanned by $x_1, \dots, x_r$ and $\wh{x}_1, \dots, \wh{x}_r$ as well as between the subspaces spanned by  $y_1, \dots, y_r$ and $\wh{y}_1, \dots, \wh{y}_r$ then it does {\em not} mean we can bound the principal angle between the subspaces spanned by
 $$x_1 \otimes y_1, \dots, x_r \otimes y_r \mbox{ and } \wh{x}_1 \otimes \wh{y}_1, \dots, \wh{x}_r \otimes \wh{y}_r$$
 However if we take all pairs of tensor products, and instead consider the principal angle between the subspaces spanned by
 $$\{x_i \otimes y_j\}_{i, j} \mbox{ and } \{\wh{x}_i \otimes \wh{y}_j\}_{i, j}$$
 then we can bound the principal angle (see Observation \ref{obs:tensorprincipalangles}). This leads to our new variant of alternating minimization, which we call {\sc Kronecker Alternating Minimization}, where the new update rule is:
 \begin{equation}\label{eq:subspacealtmin}
\{\wh{z}_{i,j}\} = \arg\min_{z_{i,j}} \norm{\left(T - \sum_{1 \leq i,j \leq r} \wh{x}_i \otimes \wh{y}_j \otimes z_{i,j}\right)\Bigg|_{S}}_2^2
\end{equation}
In particular, we are solving a least squares problem over the variables $z_{i,j}$ by taking the Kronecker product of $\wh{X}$ and $\wh{Y}$ where $\wh{X}$ is a matrix whose columns are the $\wh{x}_i$'s and similarly for $\wh{Y}$. In contrast the standard version of alternating minimization takes the Khatri-Rao product\footnote{The Khatri-Rao product, which is less famililar than the Kronecker product, takes two matrices $A$ and $B$ with the same number of columns and forms a new matrix $C$ where the $i$th column of $C$ is the tensor product of the $i$th column of $A$ and the $i$th column of $B$. This operation has many applications in tensor analysis, particularly to prove (robust) identifiability results \cite{allman2009identifiability, bhaskara2014smoothed}. }.

This modification increases the number of rank one terms in the decomposition from $r$ to $r^2$. We show that we can reduce back to $r$ without incurring too much error by finding the best rank $r$ approximation to the $n \times r^2$ matrix of the $\wh{z}_{i,j}$'s. We combine these ideas with methods for initializing alternating minimization, building on the work of Montanari and Sun \cite{spectral}, along with a post-processing algorithm to solve the non-convex exact completion problem when we are already very close to the true solution.  Our main result is:

\begin{theorem}[Informal version of Theorem~\ref{thm:main}]
Suppose $T$ is an $n \times n \times n$ low-rank, incoherent, well-conditioned tensor and its factors are robustly linearly independent. There is an algorithm that runs in nearly linear time in the number of observations and exactly completes $T$ provided that each entry is observed independently with probability $p$ where $$p \geq C \frac{r^{O(1)}}{n^{3/2}}$$ 
Moreover the constant $C$ depends polynomially on the incoherence, the condition number and the inverse of the lower bound on how far the factors are from being linearly dependent. 

\end{theorem}

\noindent We state the assumptions precisely in Section~\ref{sec:assum}.  This algorithm combines the best of many worlds:

\begin{enumerate}
    \item[(1)] It achieves \textbf{exact completion}, even when the factors are highly correlated. 
    \item[(2)] It runs in \textbf{nearly linear time} in terms of the number of observations.
    \item[(3)] The alternating minimization phase \textbf{converges at a linear rate}. 
    \item[(4)] It \textbf{scales to thousands of dimensions}, whereas previous experiments had been limited to about a hundred dimensions. 
    \item[(5)] Experimentally, in the presence of noise, it still achieves strong guarantees. In particular, it achieves \textbf{nearly optimal prediction error}. 
\end{enumerate}

\noindent We believe that our work takes an important and significant step forward to making tensor completion practical, while nevertheless maintaining strong provable guarantees.

\section{Preliminaries}

\subsection{Model and Assumptions}\label{sec:assum}

As usual in matrix and tensor completion, we will need to make an incoherence assumption as otherwise the tensor could be mostly zero and have a few large entries that we never observe.  We will also assume that the components of the tensor are not too close to being linearly dependent as otherwise the tensor will be degenerate. (Even if we fully observed it, it would not be clear how to decompose it.)

\begin{definition}\label{def:incoherence}
Given a subspace $V \subset \R^n$ of dimension $r$, we say $V$ is $\mu$-incoherent if the projection of any standard basis vector $e_i$ onto $V$ has length at most $\sqrt{\mu r/ n}$
\end{definition}

\begin{aquestion}
Consider an $n \times n \times n$ tensor with a rank $r$ CP decomposition
\[
T = \sum_{i=1}^r \sigma_i (x_i \otimes y_i \otimes z_i)
\]
where $x_i,y_i,z_i$ are unit vectors and $\sigma_1 \geq \dots \geq \sigma_r > 0$.  We make the following assumptions
\begin{itemize}
    \item {\em \textbf{ Robust Linear Independence}}: The smallest singular value of the matrix with columns given by $x_1, \dots , x_r$ is at least $c$.  The same is true for $y_1, \dots , y_r$ and $z_1, \dots , z_r$.
    \item {\em \textbf{ Incoherence}}: The subspace spanned by $x_1, \dots , x_r$ is $\mu$-incoherent.  The same is true for $y_1, \dots ,y_r$ and $z_1, \dots , z_r$.
\end{itemize}

Finally we observe each entry independently with probability $p$ and our goal is to recover the original tensor $T$.  
\end{aquestion}

We will assume that for some sufficiently small constant $\delta$, 
\[
\max\left(r, \frac{\sigma_1}{\sigma_r}, \frac{1}{c}, \mu \right) \leq n^{\delta}.
\]
In other words, we are primarily interested in how the number of observations scales with $n$, provided that it has polynomial dependence on the other parameters. 

\subsection{Technical Overview}


\paragraph{Alternating Minimization, with a Twist}

Recall the standard formulation of alternating minimization in tensor completion, given in Equation \ref{eq:naivealtmin}. Unfortunately, this approach is difficult to analyze from a theoretical perspective. In fact, in Section \ref{sec:experiments} we observe experimentally that it can indeed get stuck. Moreover, even if we add randomness by looking at a random subset of the observations in each step, it converges at a prohibitively slow rate when the factors of the tensor are correlated. Instead, in Equation \ref{eq:subspacealtmin} we gave a subtle modification to the alternating minimization steps that prevents it from getting stuck. 
We then update $\wh{z}_1, \dots , \wh{z}_r$ to be the top $r$ left singular vectors of the $n \times r^2$ matrix with columns given by the $\wh{z}_{i,j}$. With this modification, we will be able to prove strong theoretical guarantees, even when the factors of the tensor are correlated.

The main tool in the analysis of our alternating minimization algorithm is the notion of the principal angles between subspaces.   Intuitively, the principal angle between two $r$-dimensional subspaces $U,V \subset \R^n$ is the largest angle between some vector in $U$ and the subspace $V$.  For matrix completion, Jain et al. \cite{jain2013low} and Hardt \cite{matrixaltmin} analyze alternating minimization by tracking the principal angles between the subspace spanned by the top $r$ singular vectors of the true matrix and of the estimate at each step.  Our analysis follows a similar pattern.  We rely on the following key observation as the starting point for our work: 
\begin{observation}\label{obs:tensorprincipalangles}
Given subspaces $U,V \in \R^n$ of the same dimension, let $\alpha(U, V)$ be the sine of the principal angle between $U$ and $V$.  Suppose we have subspaces $U_1, V_1 \subset \R^{n_1}$ of dimension $d_1$ and $U_2, V_2 \subset \R^{n_2}$ of dimension $d_2$, then
\[
\alpha(U_1 \otimes U_2, V_1 \otimes V_2) \leq \alpha(U_1, V_1) + \alpha(U_2, V_2).
\]
\end{observation}

Thus if the subspaces spanned by the estimates $(\wh{x}_1, \dots , \wh{x}_r)$ and $(\wh{y}_1, \dots, \wh{y}_r)$ are close to the subspaces spanned by the true vectors $(x_1, \dots , x_r)$ and $(y_1, \dots , y_r)$, then the solution for $\{z_{i,j}\}$ in Equation \ref{eq:subspacealtmin} will have small error \--- i.e. $\sum_{1 \leq i,j \leq r} \wh{x}_i \otimes \wh{y}_j \otimes z_{i,j}$ will be close to $T$.  This means that the top $r$ principal components of the matrix with columns given by $\{\wh{z}_{i,j}\}$ must indeed be close to the space spanned by $z_1, \dots, z_r$.  For more details, see Corollary \ref{corollary:angle-monovariant}; this result allows us to prove that our alternating minimization steps make progress in reducing the principal angles of our subspace estimates.

On the other hand, Observation \ref{obs:tensorprincipalangles} does {\em not} hold if the tensor product is replaced with the Khatri-Rao product, which is what would be the natural route towards analyzing an algorithm that uses Equation \ref{eq:naivealtmin}.  To see this consider 
\[
U_1 = V_1 = U_2 = \begin{bmatrix} 1 & 0 \\ 0 & 1 \\ 0 & 0 \end{bmatrix}, V_2 = \begin{bmatrix} 0 & 1 \\ 1 & 0 \\ 0 & 0 \end{bmatrix}
\]
(where we really mean that $U_1,V_1,U_2,V_2$ are the subspaces spanned by the columns of the respective matrices).  Then the principal angles between $U_1,V_1$ and and $U_2,V_2$ are zero yet the principal angle between $U_1 \otimes U_2$ and $V_1 \otimes V_2$ is $\frac{\pi}{2}$ i.e. they are orthogonal.  This highlights another difficulty in analyzing Equation \ref{eq:naivealtmin}, namely that the iterates in the alternating minimization depend on the order of $(\wh{x}_1, \dots , \wh{x}_r)$ and $(\wh{y}_1, \dots, \wh{y}_r)$ and not just the subspaces themselves.

\paragraph{Initialization and Cleanup}
For the full theoretical analysis of our algorithm, in addition to the alternating minimization, we will need two additional steps.  First, we obtain a sufficiently good initialization by building on the work of Montanari and Sun \cite{spectral}.  We use their algorithm for estimating the subspaces spanned by the $\{x_i\}, \{y_i\}$ and $ \{z_i\}$.  However we then slightly perturb those estimates to ensure that our initial subspaces are incoherent. 

Next, in order to obtain exact completion rather than merely being able to estimate the entries to any desired inverse polynomial accuracy, we prove that exact tensor completion reduces to an optimization problem that is convex once we have subspace estimates that are sufficiently close to the truth (see Section \ref{sec:bitcomplexity} for a further discussion on this technicality).  To see this, applying robustness analyses of tensor decomposition \cite{moitra2018algorithmic} to our estimates at the end of the alternating minimization phase it follows that we can not only estimate the subspaces but also the entries and rank one components in the tensor decomposition to any inverse polynomial accuracy.  Now if our estimates are given by $\{\wh{\sigma_i}\},\{\wh{x_i}\},\{\wh{y_i}\}, \{\wh{z_i}\}$ we can write the expression
\begin{equation}\label{eq:informalconvexopt}
\wh{T}(\Delta) = \sum_{i=1}^r (\wh{\sigma_i} + \Delta_{\sigma_i})(\wh{x_i}  + \Delta_{x_i}) \otimes (\wh{y_i}  +\Delta_{y_i}) \otimes (\wh{z_i}  +\Delta_{z_i})
\end{equation}
and attempt to solve for $\Delta_{\sigma_i}, \Delta_{x_i}, \Delta_{y_i}, \Delta_{z_i}$ that minimize 
\[
\norm{\left(\wh{T}(\Delta) - T\right)\Big|_S}_2^2.
\]
The key observation is that since we can ensure $\{\wh{\sigma_i}\},\{\wh{x_i}\},\{\wh{y_i}\}, \{\wh{z_i}\}$ are all close to their true values, all of $\Delta_{\sigma_i}, \Delta_{x_i}, \Delta_{y_i}, \Delta_{z_i}$ are small and thus $\wh{T}(\Delta)$ can be approximated well by its linear terms.  If we only consider the linear terms, then solving for $\Delta_{\sigma_i}, \Delta_{x_i}, \Delta_{y_i}, \Delta_{z_i}$ is simply a linear least squares problem. Intuitively, $\norm{\left(\wh{T}(\Delta) - T\right)\Big|_S}_2^2$ must also be convex because the contributions from the  non-linear terms in $\wh{T}(\Delta)$ are small.  The precise formulation that we use in our algorithm will be slightly different from Equation \ref{eq:informalconvexopt} because in Equation \ref{eq:informalconvexopt}, there are certain redundancies i.e. ways to set $\Delta_{\sigma_i}, \Delta_{x_i}, \Delta_{y_i}, \Delta_{z_i}$ that result in the same $\wh{T}(\Delta)$.  See Section \ref{sec:convexopt} and in particular, Lemma \ref{lem:tractability} for details.  


\subsection{Basic Facts}
We use the following notation:
\begin{itemize}
    \item Let $U_x(T)$ denote the unfolding of $T$ into an $n \times n^2$ matrix where the dimension of length $n$ corresponds to the $x_i$.  Define $U_y(T), U_z(T)$ similarly.
    \item  Let $V_x$ be the subspace spanned by $x_1, \dots , x_r$ and define $V_y,V_z$ similarly.
    \item Let $M_x$ be the matrix whose columns are $x_1, \dots , x_r$ and define $M_y, M_z$ similarly. 
\end{itemize}

The following claim, which states that the unfolded tensor $U_x(T)$ is not a degenerate matrix, will be used repeatedly later on.

\begin{claim}\label{claim:nondegenerate}
The $r$\ts{th} largest singular value of $U_x(T)$ is at least $c^3\sigma_r$
\end{claim}
\begin{proof}
  Let $D$ be the $r \times r$ diagonal matrix whose entries are $\sigma_1, \dots , \sigma_r$.  Let $N$ be the $r \times n^2$ matrix whose rows are $y_1 \otimes z_1, \dots , y_r \otimes z_r$ respectively.  Then
\[
U_x(T) = M_xDN
\]
For any unit vector $v \in V_x$, $||vM_x||_2 \geq c$.  Also $N$ consists of a subset of the rows of $M_y \otimes M_z$ so the smallest singular value of $N$ is at least $c^2$.  Thus 
\[
||vU_x(T)||_2 \geq c^3\sigma_r.
\]
Since $V_x$ has dimension $r$, this implies that the $r$\ts{th} largest singular value of $U_x(T)$ is at least $c^3\sigma_r$.
\end{proof}

A key component of our analysis will be tracking principal angles between subspaces.  Intuitively, the principal angle between two $r$-dimensional subspaces $U,V \subset \R^n$ is the largest angle between some vector in $U$ and the subspace $V$.
\begin{definition}\label{def:principalangle}
For two subspaces $U,V \subset \R^n$ of dimension $r$, we let $\alpha(U,V)$ be the sine of the principal angle between $U$ and $V$.  More precisely, if $U$ is a $n \times r$ matrix whose columns form an orthonormal basis for $U$ and $V_{\perp}$ is a $n \times (n-r)$ matrix whose columns form an orthonormal basis for the orthogonal complement of $V$, then 
\[
\alpha(U,V) = \norm{V_{\perp}^T U}_{\op}
\]
\end{definition}

\begin{observation}[Restatement of Observation \ref{obs:tensorprincipalangles}]
Given subspaces $U_1, V_1 \subset \R^{n_1}$ of dimension $d_1$ and $U_2, V_2 \subset \R^{n_2}$ of dimension $d_2$, we have
\[
\alpha(U_1 \otimes U_2, V_1 \otimes V_2) \leq \alpha(U_1, V_1) + \alpha(U_2, V_2)
\]
\end{observation}
\begin{proof}
We slightly abuse notation and use $U_1,V_1, U_2,V_2$ to denote matrices whose columns form an orthonormal basis of the respective subspaces.  Note that the cosine of the principal angle between $U_1$ and $V_1$ is equal to the smallest singular value of $U_1^TV_1$ and similar for $U_2$ and $V_2$.  Next note 
\[
(U_1 \otimes U_2)^T (V_1 \otimes V_2) = (U_1^TV_1) \otimes (U_2^TV_2).
\]
Thus 
\[
1 - \alpha(U_1 \otimes U_2, V_1 \otimes V_2)^2 = (1 -\alpha(U_1, V_1)^2)(1 - \alpha(U_2, V_2)^2)
\]
and we conclude
\[
\alpha(U_1 \otimes U_2, V_1 \otimes V_2) \leq \alpha(U_1, V_1) + \alpha(U_2, V_2)
\]

\end{proof}

In the analysis of our algorithm, we will also need to understand the incoherence of the tensor product of vector spaces.  The following claim gives us a simple relation for this. 

\begin{claim}\label{claim:incoherencebound}
Suppose we have subspaces $V_1 \subset \R^{n_1}$ and $V_2 \subset \R^{n_2}$ with dimension $r_1, r_2$ that are $\mu_1$ and $\mu_2$ incoherent respectively.  Then $V_1 \otimes V_2$ is $\mu_1\mu_2$-incoherent.
\end{claim}
\begin{proof}
Let $M_1,M_2$ be matrices whose columns are orthonormal bases for $V_1,V_2$ respectively.  Then the columns of $M_1 \otimes M_2$ form an orthonormal basis for $V_1 \otimes V_2$.  All rows of $M_1$ have norm at most $\sqrt{\frac{\mu_1 r_1}{n_1}}$ and all rows of $M_2$ have norm at most $\sqrt{\frac{\mu_2 r_2}{n_2}}$ so thus all rows of $M_1 \otimes M_2$ have norm at most $\sqrt{\frac{\mu_1 \mu_2 r_1r_2}{n_1n_2}}$ and we are done.
\end{proof}

\subsection{Matrix and Incoherence Bounds}\label{sec:general-ineqs}
Here, we prove a few general results that we will use later to bound the principal angles and incoherence of the subspace estimates at each step of our algorithm.

\begin{claim}\label{claim:row-incoherence}
Let $X$ be an $n \times m$ matrix such that: 
\begin{itemize}
    \item The rows of $X$ have norm at most $c \sqrt{\frac{r}{n}}$
    \item The $r$\ts{th} singular value of $X$ is at least $\rho$
\end{itemize}
Then the subspace spanned by the top $r$ left singular vectors of $X$ is $\frac{c^2 r}{\rho^2}$-incoherent.
\end{claim}
\begin{proof}
Let the top $r$ left singular vectors of $X$ be $v_1, \dots , v_r$.  There are vectors $u_1, \dots , u_r$ such that $v_i^T = Xu_i$ for all $i$.  Furthermore, we can ensure
\[
||u_i||_2 \leq \frac{1}{\rho}
\]
for all $i$.  Now for a standard basis vector say $e_j$, its projection onto the subspace spanned by $v_1, \dots , v_r$ has norm
\[
\sqrt{(e_j^TXu_1)^2 + \dots + (e_j^TXu_r)^2} \leq \sqrt{r} \cdot c\sqrt{\frac{r}{n}}\frac{1}{\rho}. 
\]
Thus, the column space of $X$ is $\frac{c^2 r}{\rho^2}$-incoherent.
\end{proof}

\begin{claim}\label{claim:operator-anglebound}
Let $A,B$ be $n \times m$ matrices with rank $r$.  Let $\delta = \norm{A - B}_{\op}$.  Assume that the $r$\ts{th} singular value of $B$ is at least $\rho$.   Then the sine of the principal angle between the subspaces spanned by the columns of $A$ and $B$ is at most $\frac{\delta}{\rho}$.
\end{claim}
\begin{proof}
Let $V_A$ be the column space of $A$ and $V_B$ be the column space of $B$.  Let $V_{\perp}$ be a $(n-r) \times n$ matrix whose rows form an orthonormal basis of the orthogonal complement of $V_A$.  Note that 
\[
\norm{V_{\perp}B} \geq \rho \alpha(V_A, V_B)
\]
Now $\norm{V_{\perp}B} \leq \norm{V_{\perp}(B - A)} \leq \delta$. 
Thus $\alpha(V_A, V_B) \leq \frac{\delta}{\rho}$ which completes the proof.
\end{proof}

\subsection{Sampling Model}
In our sampling model, we observe each entry of the tensor $T$ independently with probability $p$.  In our algorithm we will require splitting the observations into several independent samples.  To do this, we rely on the following claim:
\begin{claim}\label{claim:samplesplitting}
Say we observe a sample $\wh{T}$ where every entry of $T$ is revealed independently with probability $p$.  Say $p_1 + p_2 \leq p$.  We can construct two independent samples $\wh{T_1}, \wh{T_2}$ where in the first sample, every entry is observed independently with probability $p_1$ and in the second, every entry is observed independently with probability $p_2$.
\end{claim}
\begin{proof}
For each entry in $\wh{T}$ that we observe, reveal it in only $\wh{T_1}$ with probability $\frac{p_1 - p_1p_2}{p}$, reveal it in only $\wh{T_2}$ with probability $\frac{p_2 - p_1p_2}{p}$ and reveal it in both with probability $\frac{p_1p_2}{p}$.  Otherwise, don't reveal the entry in either $\wh{T_1}$ or $\wh{T_2}$.  It can be immediately verified that $\wh{T_1}$ and $\wh{T_2}$ constructed in this way have the desired properties.
\end{proof}

\section{Our Algorithms}\label{sec:algorithm}

Here we will give our full algorithm for tensor completion, which consists of the three phases discussed earlier, initialization through spectral methods, our new variant of alternating minimization, and finally a post-processing step to recover the entries of the tensor exactly (when there is no noise). In Section \ref{sec:lineartime}, we will show that our algorithm can be implemented so that it runs in nearly linear time in terms of the number of observations.

\begin{algorithm}[H]
\caption{{\sc Full Exact Tensor Completion} }
\begin{algorithmic} 
\State \textbf{Input:} Let $\wh{T}$ be a sample where we observe each entry of $T$ independently with probability $p$
\State Let $p_1,p_2,p_3$ be parameters with $p_1 + p_2 + p_3 \leq p$
\State Split $\wh{T}$ into three samples $\wh{T_1}, \wh{T_2}, \wh{T_3}$ using Claim \ref{claim:samplesplitting} such that
\begin{itemize}
    \item In $\wh{T_1}$ each entry is observed with probability $p_1$
    \item In $\wh{T_2}$ each entry is observed with probability $p_2$
    \item In $\wh{T_3}$ each entry is observed with probability $p_3$
\end{itemize}
\State Run {\sc Initialization} using $\wh{T_1}$ to obtain initial estimates $V_x^0, V_y^0, V_z^0$ for the subspaces $V_x,V_y,V_z$
\State Run {\sc Kronecker Alternating Minimization} using $\wh{T_2}$ and initial subspace estimates $V_x^0, V_y^0, V_z^0$ to obtain refined subspace estimates $\wh{V_x}, \wh{V_y}, \wh{V_z}$
\State Run {\sc Post-Processing via Convex Optimization} using $\wh{T_3}$ and the subspace estimates $\wh{V_x}, \wh{V_y}, \wh{V_z}$ from the  previous step.

\end{algorithmic}
\end{algorithm}

Our initialization algorithm is based on \cite{spectral}.  For ease of notation, we make the following definition:  
\begin{definition}
Let $\Pi$ be the projection map that projects an $n \times n$ matrix onto its diagonal entries and let $\Pi_{\perp}$ denote the projection onto the orthogonal complement.
\end{definition}

\begin{algorithm}[H]
\caption{{\sc Initialization} }
\begin{algorithmic} 
\State \textbf{Input:} Let $\wh{T}$ be an input tensor where each entry is observed with probability $p_1$
\State Let $U = U_x(\wh{T})$ be the unfolded tensor with all unobserved entries equal to $0$

\State Define
\[
\wh{B} = \frac{1}{p_1} \Pi(UU^T) + \frac{1}{p_1^2}\Pi_{\perp}(UU^T)
\]
\State Let $X$ be the matrix whose columns are given by the top $r$ singular vectors of $\wh{B}$ 
\State Zero out all rows of $X$ that have norm at least $\tau \sqrt{\frac{r}{n}}$ where $\tau =  \left(\frac{2\mu r}{c^2} \cdot \frac{\sigma_1^2}{\sigma_r^2}\right)^{5}$
and let $X_0$ be the resulting matrix
\State Let $V_x^0$ be the subspace spanned by the columns of $X_0$
\State Compute $V_y^0, V_z^0$ similarly using the corresponding unfoldings of $\wh{T}$
\State \textbf{Output:} $V_x^0,V_y^0, V_z^0 $
\end{algorithmic}
\end{algorithm}


\begin{algorithm}[H]
\caption{{\sc Kronecker Alternating Minimization} }
\begin{algorithmic} 
\State \textbf{Input:} Let $\wh{T}$ be a sample where we observe each entry of $T$ independently with probability $p_2$
\State \textbf{Input:} Let $V_x^0, V_y^0, V_z^0$ be initial subspace estimates that we are given
\State Set $k = 10^2\log \frac{n\sigma_1}{c\sigma_r}$
\State Let $p'$ be a parameter such that $kp' \leq p_2$
\State Split $\wh{T}$ into independent samples $\wh{T_1}, \dots \wh{T_{k}}$ using Claim \ref{claim:samplesplitting} such that 
\begin{itemize}
    \item For $1 \leq i \leq k$, in $\wh{T_i}$ each entry is revealed with probability $p'$
\end{itemize}

\For  {$t = \{0,1, \dots , k-1 \}$ }
\State Let $B_t$ be an $r^2 \times n^2$ matrix whose rows are an orthonormal basis for $V_y^{t} \otimes V_z^{t}$

\State Consider the sample $\wh{T}_{t+1}$ and let $S_{t+1}$ be the set of observed entries
\State Let $H_{t+1}$ be the solution to $\min_H\norm{(U_x(\widehat{T}_{t+1}) - HB_t)\big|_{S_{t+1}}}_2^2$

\State Let $V_x^{t+1}$ be the space spanned by the top $r$ left singular vectors of $H_{t+1}$
\State Compute $V_y^{t+1}$ from $V_x^t, V_z^t$ and compute $V_z^{t+1}$ from $V_x^t, V_y^t$ similarly
\EndFor
\State \textbf{Output:} $V_x^{k}, V_y^{k}, V_z^{k}$

\end{algorithmic}
\end{algorithm}

\begin{algorithm}[H]
\caption{{\sc Post-Processing via Convex Optimization} }
\begin{algorithmic} 
\State  \textbf{Input:} Let $\wh{T}$ be a sample where each entry is observed with probability $p_3$
\State \textbf{Input:} Let $\wh{V_x}, \wh{V_y}, \wh{V_z}$ be subspace estimates that we are given
\State Split $\wh{T}$ into two independent samples $\wh{T_{-}}, \wh{T_{\sim}}$ where each entry is observed with probability $p_3/2$

\State Let $S$ be the set of observed entries in $\wh{T_{-}}$

\State Define:
\[
T' = \arg\min_{T' \in \wh{V_x} \otimes \wh{V_y} \otimes \wh{V_z}}\norm{(T' - \wh{T_{-}})|_S}_2^2
\]
\State Run Jennrich's algorithm (see Section \ref{sec:jennrichalg}) to decompose $T'$ into $r$ rank-$1$ components 
\[
T' = T_1 + \dots  + T_r
\]
\State For each $1 \leq i \leq r$, write $T_i = \wh{\sigma_i}\wh{x_i} \otimes \wh{y_i} \otimes \wh{z_i}$ where $\wh{x_i}, \wh{y_i}, \wh{z_i}$ are unit vectors and $\wh{\sigma_i} \geq 0$.

\State Let $S_{\sim}$ be the set of observed entries in $\wh{T_{\sim}}$
\State Solve the following constrained optimization problem where $a_1,b_1, c_1, \dots , a_r, b_r, c_r \in \R^n$:
\[
\min_{a_i,b_i,c_i}\norm{\left(T - \sum_{i=1}^r (\wh{\sigma_i}(\wh{x_i} + a_i)) \otimes (\wh{y_i} + b_i) \otimes (\wh{z_i} + c_i)\right)\Bigg|_{S_{\sim}}}_2^2
\]
over the polytope $ Q(\wh{x_1}, \dots , \wh{x_r}, \wh{y_1}, \dots, \wh{y_r}, \wh{z_1}, \dots, \wh{z_r}) $ (defined below).

\State \textbf{Output:}
\[
T_{\textsf{est}}= \sum_{i=1}^r (\wh{\sigma_i}(\wh{x_i} + a_i)) \otimes (\wh{y_i} + b_i) \otimes (\wh{z_i} + c_i)
\]
\end{algorithmic}
\end{algorithm}

\begin{definition}[Definition of $ Q $]
For each $1 \leq i \leq r$, let $y_i'$ be the unit vector in $\spn(\wh{y_1}, \dots , \wh{y_r})$ that is orthogonal to $\wh{y_1}, \dots , \wh{y_{i-1}}, \wh{y_{i+1}}, \dots , \wh{y_r}$ and define $z_i'$ similarly.  Let $$Q(\wh{x_1}, \dots , \wh{x_r}, \wh{y_1}, \dots, \wh{y_r}, \wh{z_1}, \dots, \wh{z_r}) $$ be the polytope consisting of all $\{a_1,b_1,c_1, \dots, a_r, b_r, c_r\}$ such that:
\begin{itemize}
    \item $0 \leq ||a_i||_{\infty},||b_i||_{\infty}, ||c_i||_{\infty} \leq \left(\frac{c\sigma_r}{10n\sigma_1}\right)^{10}$ for all $1 \leq i \leq r$.
    \item $b_i \cdot y_i' = 0$ and $c_i \cdot z_i' = 0$ for all $1 \leq i \leq r$.
\end{itemize}
\end{definition}

\begin{remark}
Note we will prove that the constrained optimization problem is strongly convex (and thus can be solved efficiently) in Section \ref{sec:convexopt}.
\end{remark}

We will now state our main theorem:

\begin{theorem} \label{thm:main}
The {\sc Full Exact Tensor Completion} algorithm run with the following parameter settings:
\begin{align*}
&p = 2\left(\frac{\mu r \log n  }{c} \cdot \frac{\sigma_1}{\sigma_r} \right)^{300}\frac{1}{n^{3/2}} \\
&p_1 = \left(\frac{2\mu r \log n}{c^2} \cdot \frac{\sigma_1^2}{\sigma_r^2}\right)^{10}\frac{1}{n^{3/2}} \\
&p_2 = \left(\frac{\mu r \log n  }{c} \cdot \frac{\sigma_1}{\sigma_r} \right)^{300}\frac{1}{n^{3/2}} \\
&p_3 = 2\frac{\log^2 n}{n^2} \left(\frac{10r\mu}{ c} \cdot \frac{\sigma_1}{\sigma_r}\right)^{10}
\end{align*}
successfully outputs $T$ with probability at least $0.9$.  Furthermore, the algorithm can be implemented to run in $ n^{3/2}\poly(r,\log n, \sigma_1/\sigma_r, \mu, 1/c)$ time.

\end{theorem}

\subsection{``Exact" Completion and Bit Complexity}\label{sec:bitcomplexity}

Technically, exact completion only makes sense when the entries of the hidden tensor have bounded bit complexity.  Our algorithm achieves exact completion in the sense that, if we assume that all of the entries of the hidden tensor have bit complexity $B$, then the number of observations that our algorithm requires \emph{does not} depend on $B$ while the runtime of our algorithm depends polynomially on $B$.  This is the strongest possible guarantee one could hope for in the Word RAM model.  Note that the last step of our algorithm involves solving a convex program.  If we assume that the entries of the original tensor have bounded bit complexity then it suffices to solve the convex program to sufficiently high precision \cite{grotschel2012geometric} and then round the solution.

\section{Outline of Proof of Theorem \ref{thm:main}}
Here we give an outline of the proof of Theorem \ref{thm:main} as many of the parts will be deferred to the appendix. The first step involves proving that with high probability, the {\sc Initialization} algorithm outputs subspaces that are incoherent and  have constant principal angle with the true subspaces spanned by the unknown factors.  The proof of the following theorem is deferred to Section \ref{sec:initialization}.

\begin{theorem}\label{thm:init}
With probability $1 - \frac{1}{n^{10}}$, when the {\sc Initialization} algorithm is run with
\[
p_1 = \left(\frac{2\mu r \log n}{c^2} \cdot \frac{\sigma_1^2}{\sigma_r^2} \right)^{10}\frac{1}{n^{3/2}}, 
\]
the output subspaces $V_x^0, V_y^0, V_z^0$ satisfy
\begin{itemize}
    \item $\max\left(\alpha(V_x, V_x^{0}),\alpha(V_y, V_y^{0}), \alpha(V_z, V_z^{0}) \right) \leq 0.1$
    \item The subspaces $V_x^0, V_y^0, V_z^0$ are $\mu'$-incoherent where $\mu' = \left(\frac{2\mu r}{c^2} \cdot \frac{\sigma_1^2}{\sigma_r^2}\right)^{10}$
\end{itemize}
\end{theorem}

Next, we prove that each iteration of alternating minimization decreases the principal angles between our subspace estimates and the true subspaces. This is our main contribution. 

\begin{theorem}\label{thm:alternatingmin}
Consider the {\sc Kronecker Alternating Minimization} algorithm. Fix a timestep $t$ and assume that the subspaces $V_x^t, V_y^t, V_z^t$ corresponding to the current estimate are $\mu'$-incoherent where $\mu' = \left(\frac{2\mu r}{c^2} \cdot \frac{\sigma_1^2}{\sigma_r^2}\right)^{10}$.  Also assume 
\[
\max\left(\alpha(V_x, V_x^{t}),\alpha(V_y, V_y^{t}), \alpha(V_z, V_z^{t}) \right) \leq 0.1.
\]
If $p' \geq \frac{\log^2((n\sigma_1)/(c\sigma_r))}{n^2} \left(\frac{10r\mu'}{ c} \cdot  \frac{\sigma_1}{\sigma_r}\right)^{10}$ then after the next step, with probability $1 - \frac{1}{10^4 \log ((n\sigma_1)/(c\sigma_r))}$
\begin{itemize}
\item $\max\left(\alpha(V_x, V_x^{t+1}),\alpha(V_y, V_y^{t+1}), \alpha(V_z, V_z^{t+1}) \right) \leq 0.2 \max\left(\alpha(V_x, V_x^t),\alpha(V_y, V_y^t), \alpha(V_z, V_z^t)\right)$
\item  The subspaces $V_x^{t+1}, V_y^{t+1}, V_z^{t+1}$ are $\mu'$-incoherent 
\end{itemize}
\end{theorem}
\noindent This theorem is proved in Section \ref{sec:altmin}

In light of the previous two theorems, by running a logarithmic number of iterations of alternating minimization, we can estimate the subspaces $V_x,V_y,V_z$ to within any inverse polynomial accuracy.  This implies that we can estimate the entries of the true tensor $T$ to within any inverse polynomial accuracy.  A robust analysis of Jennrich's algorithm implies that we can then estimate the rank one components of the true tensor to within any inverse polynomial accuracy.  Finally, since our estimates for the parameters $\wh{\sigma_i},\wh{x_i},\wh{y_i},\wh{z_i}$ are close to the true parameters, we can prove that the optimization problem formulated in the {\sc Post-Processing via Convex Optimization} algorithm is smooth and strongly convex.  See Section \ref{sec:projectionanddecomp} and Section \ref{sec:convexopt} for details.

\section{Alternating Minimization}\label{sec:altmin}

This section is devoted to the proof of Theorem \ref{thm:alternatingmin}. 


\subsection{Least Squares Optimization}

We will need to analyze the least squares optimization in the alternating minimization step of our {\sc Kronecker Alternating Minimization} algorithm.  We use the following notation.
\begin{itemize}
    \item Let the rows of $H_{t+1}$ be $r_1, \dots , r_n$. 
    \item We will abbreviate  $U_x(T)$ with $U_x$.  Let the rows of $U_x(T)$ be $u_1, \dots , u_n$.
    \item For each $1 \leq j \leq n$, let $P_j \subset [n^2]$ be the set of indices that are revealed in the $j$\ts{th} row of $U_x(\widehat{T}_{t+1})$.
    \item Let $\Pi_j$ be the $n^2 \times n^2$ matrix with one in the diagonal entries corresponding to elements of $P_j$ and zero everywhere else.
\end{itemize}
Let $E = H_{t+1} - U_xB_t^T$ and let its rows be $s_1, \dots , s_n$.  Note $U_xB_t^T$ is the solution to the least squares optimization problem if we were able to observe all the entries i.e. if the optimization problem were not restricted to the set $S_{t+1}$.  Thus, we can think of $E$ as the error that comes from only observing a subset of the entries.  The end goal of this section is to bound the Frobenius norm of $E$.  
\\\\
First we will need a technical claim that follows from a matrix Chernoff bound.

\begin{claim}\label{claim:matrix_inv}
If $p' \geq \frac{\log^2 n}{n^2} \left(\frac{10r\mu'}{ c} \cdot \frac{\sigma_1}{\sigma_r}\right)^{10}$ then With at least $1 -  \left(\frac{c\sigma_r}{n\sigma_1} \right)^{20}$ probability 
\[
||(B_t\Pi_jB_t^T)^{-1}|| \leq \frac{2}{ p'}
\]
for all $1 \leq j \leq n$.
\end{claim}
\begin{proof}
Note by Claim \ref{claim:incoherencebound}, the columns of $B_t$ have norm at most $\frac{\mu' r}{n}$.  Consider $\frac{1}{p'}B_t\Pi_jB_t^T$.  This is a sum of independent rank $1$ matrices with norm bounded by $\frac{\mu'^2r^2}{p'n^2}$ and the expected value of the sum is $I$, the identity matrix.  Thus by Claim \ref{claim:matrix-chernoff}
\[
\Pr\left[\lambda_{\min}\left(\frac{1}{p'}B_t\Pi_jB_t^T \right) \leq \frac{1}{2}  \right] \leq ne^{-\frac{p'n^2}{8\mu'^2r^2}} \leq \left(\frac{c \sigma_r}{n\sigma_1}\right)^{25}
\]
\end{proof}

The claim below follows directly from writing down the explicit formula for the solution to the least squares optimization problem.
\begin{claim}
We have for all $1 \leq j \leq n$
\[
r_j = u_j\Pi_jB_t^T(B_t \Pi_j B_t^T)^{-1}
\]
\end{claim}
\begin{proof}
Note $r_j$ must have the property that the vector $u_j - r_jB_t$ restricted to the entries indexed by $P_j$ is orthogonal to the space spanned by the rows of $B_t$ restricted to the entries indexed by $P_j$.  This means
\[
(u_j - r_jB_t)(B_t\Pi_j)^T = 0
\]
The above rearranges as 
\[
r_j (B_t\Pi_jB_t^T) = u_j\Pi_jB_t^T
\]
from which we immediately obtain the desired.
\end{proof}

After some direct computations, the previous claim implies:
\begin{claim}
We have for all $1 \leq j \leq n$
\[
s_j = u_j(I - B_t^TB_t)\Pi_jB_t^T(B_t \Pi_j B_t^T)^{-1}
\]
\end{claim}
\begin{proof}
Note
\[
s_j = r_j - u_jB_t^T = u_j\Pi_jB_t^T(B_t \Pi_j B_t^T)^{-1} - u_jB_t^T = u_j(I - B_t^TB_t)\Pi_jB_t^T(B_t \Pi_j B_t^T)^{-1}
\]
\end{proof}

Now we are ready to prove the main result of this section.

\begin{lemma}\label{lemma:frobeniusbound}
With probability at least $1 - \frac{1}{10^5 \log ((n\sigma_1)/ (c\sigma_r))}$
\[
||E||_2 \leq \frac{10^3  \mu' r^3 \sigma_1 \sqrt{\log((n\sigma_1)/ (c\sigma_r))}}{n\sqrt{p'}} \left(\alpha(V_y, V_y^t) + \alpha(V_z, V_z^t) \right)
\]
\end{lemma}
\begin{proof}
Note $B_tB_t^T = I$ since its rows form an orthonormal basis so $u_j (I - B_t^TB_t)B_t^T = 0$.  Thus
\[
s_j = u_j(I - B_t^TB_t)(\Pi_j - p'I)B_t^T(B_t \Pi_j B_t^T)^{-1}
\]
Let $q_j = u_j(I - B_t^TB_t)$.  By Observation \ref{obs:tensorprincipalangles},
\begin{equation}\label{eq:anglebound}
||q_j||_2 \leq ||u_j||_2\alpha(V_y \otimes V_z, V_y^t \otimes V_z^t ) \leq ||u_j||_2\left(\alpha(V_y, V_y^t) + \alpha(V_z, V_z^t) \right)
\end{equation}

Now we upper bound $||q_j(\Pi_j - p'I)B_t^T ||$.  Let the columns of $B_t^T$ be $c_1, \dots , c_{r^2}$ respectively.  Note that since the entries of $\Pi_j - p'I$ have expectation $0$ and variance $p'(1-p')$
\[
\E[||q_j(\Pi_j - p'I)c_i||_2^2] \leq ||q_j||_2^2\cdot p'(1-p') \cdot ||c_i||_{\infty}^2 \leq \frac{p'||q_j||_2^2 \mu'^2 r^2}{n^2}
\]
Thus
\begin{equation}\label{eq:expectationbound}
\E[||q_j(\Pi_j - p'I)B_t^T||_2^2] \leq \frac{p'||q_j||_2^2 \mu'^2 r^4}{n^2}
\end{equation}
Now by Claim \ref{claim:matrix_inv}, with probability $1 - \left(\frac{c\sigma_r}{n\sigma_1} \right)^{20}$ we have 
\[
||(B_t\Pi_jB_t^T)^{-1}|| \leq \frac{2}{ p'}
\]
for all $1 \leq j \leq n$.  Denote this event by $\Gamma$.  Let 
\[
\E_{\Gamma}\left[||E||_2^2\right] = \Pr\left[\Gamma\right]\E\left[||E||_2^2 \: \big| \: \Gamma\right]
\]
In other words, $\E_{\Gamma}[||E||_2^2]$ is the sum of the contribution to $\E\left[||E||_2^2\right]$  from events where $\Gamma$ happens. 
\\\\
Using (\ref{eq:anglebound}) and (\ref{eq:expectationbound}) we have
\begin{align*}
\E_{\Gamma}\left[||E||_2^2 \right] = \sum_{j=1}^n \E_{\Gamma}\left[||s_j||_2^2\right] \leq \frac{4}{p'^2}\frac{p' \mu'^2 r^4}{n^2}\left(\sum_{j=1}^n ||q_j||_2^2\right) \\ \leq \frac{4\mu'^2 r^4}{p'n^2}\left(\alpha(V_y, V_y^t) + \alpha(V_z, V_z^t) \right)^2\sum_{j=1}^n ||u_j||_2^2 \\ \leq \frac{4\mu'^2 r^6 \sigma_1^2}{p'n^2}\left(\alpha(V_y, V_y^t) + \alpha(V_z, V_z^t) \right)^2
\end{align*}

By Markov's inequality, the probability that $\Gamma$ occurs and \[
||E||_2 \geq \frac{10^3 \mu' r^3 \sigma_1\sqrt{\log((n\sigma_1 )/ (c\sigma_r))}}{n\sqrt{p'}} \left(\alpha(V_y, V_y^t) + \alpha(V_z, V_z^t) \right)
\]
is at most $\frac{1}{2\cdot 10^5 \log ((n\sigma_1)/ (c\sigma_r))}$.  Thus with probability at least
\[
1 - \frac{1}{2 \cdot 10^5 \log((n\sigma_1)/ (c\sigma_r))} -  \left(\frac{c\sigma_r}{n\sigma_1} \right)^{20} \geq 1 - \frac{1}{10^5 \log ( (n\sigma_1)/ (c\sigma_r))}
\]
we have
\[
||E||_2 \leq \frac{10^3 \mu' r^3\sigma_1 \sqrt{\log((n\sigma_1)/ (c\sigma_r))}}{n\sqrt{p'}} \left(\alpha(V_y, V_y^t) + \alpha(V_z, V_z^t) \right)
\]
\end{proof}

We will need one additional lemma to show that the incoherence of the subspaces is preserved.  This lemma is an upper bound on the norm of the rows of $E$.  Note this upper bound is much weaker than the one in the previous lemma and does not capture the progress made by alternating minimization but rather is a fixed upper bound that holds for all iterations.
\begin{lemma}\label{lemma:rowbound}
If $p' \geq \frac{\log^2 ((n\sigma_1)/ (c\sigma_r))}{n^2} \left(\frac{10r\mu'}{ c} \cdot \frac{\sigma_1}{\sigma_r}\right)^{10}$ then with probability at least $ 1- \left(\frac{c\sigma_r}{n\sigma_1} \right)^{15}$ all rows of $E$ have norm at most $\frac{0.1c^3\sigma_r}{\sqrt{n}}$
\end{lemma}
\begin{proof}
We use the same notation as the proof of the previous claim.  Fix indices $1 \leq j \leq n,1 \leq i \leq r^2$.  It will suffice to obtain a high probability bound for 
\[
|q_j(\Pi_j - p'I)c_i |
\]
and then union bound over all indices.  Recall
\[
q_j = u_j - u_jB_t^TB_t
\]
and by our incoherence assumptions, all entries of $u_j$  have magnitude at most $r\sigma_1\left(\frac{\mu r}{n}\right)^{3/2}$.  By Claim \ref{claim:incoherencebound}, all entries of $u_jB_t^TB_t$ have magnitude at most 
\[
||u_jB_t^T||_2\frac{\mu' r}{n} \leq ||u_j||_2\frac{\mu' r}{n} \leq r\sigma_1\sqrt{\frac{\mu r}{n}}\frac{\mu' r}{n}
\]
Thus all entries of $q_j$ have magnitude at most $2r\sigma_1\sqrt{\frac{\mu r}{n}}\frac{\mu' r}{n}$.  Let $\tau$ be the vector obtained by taking the entrywise product of $q_j$ and $c_i$ and say its entries are $\tau_1, \dots , \tau_{n^2}$.  Note that by Claim \ref{claim:incoherencebound} the entries of $c_i$ are all at most $\frac{\mu' r}{n}$.  Thus the entries of $\tau$ are all at most
\[
\gamma = 2r\sigma_1\sqrt{\frac{\mu r}{n}}\frac{\mu'^2 r^2}{n^2}
\]

Next observe that $q_j(\Pi_j - p'I)c_i$ is obtained by taking a sum $\epsilon_1\tau_1 + \dots + \epsilon_{n^2}\tau_{n^2}$ where $\epsilon_1, \dots , \epsilon_{n^2}$ are sampled independently and are equal to $-p'$ with probability $1 - p'$ and equal to $1-p'$ with probability $p'$.  We now use Claim \ref{claim:modifiedchernoff}.  Note $p' \geq \frac{1}{n^2}$ clearly.  Thus
\[
\Pr\left[ \: |q_j(\Pi_j - p'I)c_i| \geq  10^3 \log ((n\sigma_1)/(c\sigma_r)) n\gamma\sqrt{p'} \right] \leq \left(\frac{c\sigma_r}{n\sigma_1}\right)^{25}.
\]
Note
\[
n\gamma \sqrt{p'} = \frac{2r^{3.5}\mu'^2\mu^{0.5}\sqrt{p'}\sigma_1}{n^{3/2}}
\]
Thus
\[
\Pr\left[ \: |q_j(\Pi_j - p'I)c_i| \leq \frac{2\cdot 10^3 \log ((n\sigma_1)/(c\sigma_r))r^{3.5}\mu'^2\mu^{0.5}\sqrt{p'}\sigma_1}{n^{3/2}} \right] \geq 1 - \left(\frac{c\sigma_r}{n\sigma_1}\right)^{25}.
\]

Union bounding over $1 \leq i \leq r^2$ implies that with probability at least $1 -  \left(\frac{c\sigma_r}{n\sigma_1} \right)^{23}$
\[
||q_j(\Pi_j - p'I)B_t^T|| \leq  \frac{2 \cdot 10^3 \log ((n\sigma_1)/(c\sigma_r))r^{4.5}\mu'^2\mu^{0.5}\sqrt{p'}\sigma_1}{n^{3/2}}
\]
Finally, since
\[
s_j = u_j(I - B_t^TB_t)(\Pi_j - p'I)B_t^T(B_t \Pi_j B_t^T)^{-1} = q_j(\Pi_j - p'I)B_t^T(B_t \Pi_j B_t^T)^{-1}
\]
combining with Claim \ref{claim:matrix_inv} and union bounding over all $1 \leq j \leq n$, we see that with at least $1 -  \left(\frac{c\sigma_r}{n\sigma_1} \right)^{15}$ probability, all rows of $E$ have norm at most
\[
 \frac{4 \cdot 10^3 \log ((n\sigma_1)/(c\sigma_r))r^{4.5}\mu'^2\mu^{0.5}\sigma_1}{n^{3/2}\sqrt{p'}} \leq \frac{0.1c^3\sigma_r}{\sqrt{n}}
\]
which completes the proof. 
\end{proof}

\subsection{Progress Measure}
Now we will use the bounds in Lemma \ref{lemma:frobeniusbound} and Lemma \ref{lemma:rowbound} on the error term $E$ to bound the principal angle with respect to $V_x$ and the incoherence of the new subspace estimate $V_x^{t+1}$.  This will then complete the proof of Theorem \ref{thm:alternatingmin}.  We first need a preliminary  result controlling the $r$\ts{th} singular value  of  $U_xB_t^T$.  Note this is necessary because if the $r$\ts{th} singular value of $U_xB_t^T$ were too small, then it could be erased by the error term $E$.

\begin{claim}\label{claim:nondegenerate2}
The $r$\ts{th} largest singular value of $U_xB_t^T$ is at least $\left(1 - \alpha(V_y, V_y^t) - \alpha(V_z, V_z^t)\right)c^3\sigma_r$
\end{claim}
\begin{proof}
By Claim \ref{claim:nondegenerate}, the $r$\ts{th} largest singular value of $U_x$ is at least $c^3\sigma_r$.  Therefore there exists an $r$-dimensional subspace of $\R^n$, say $V$ such that for any unit vector $v \in V$, $||vU_x|| \geq c^3\sigma_r$.  Now for any vector $u$ in $V_y \otimes V_z$, 
\[
||uB_t^T|| \geq \sqrt{1 - \alpha(V_y \otimes V_z, V_y^t \otimes V_z^t)^2}||u||
\]
Next $vU_x$ is contained in the row span of $U_x$ which is contained in $V_y \otimes V_z$.  Thus for any unit vector $v \in V$
\[
||vU_xB_t^T|| \geq  c^3\sigma_r\sqrt{1 - \alpha(V_y \otimes V_z, V_y^t \otimes V_z^t)^2} \geq \left(1 - \alpha(V_y, V_y^t) - \alpha(V_z, V_z^t)\right)c^3\sigma_r
\]
\end{proof}

Now we can upper bound the principal angle between $V_x$ and $V_x^{t+1}$ in terms of the principal angles for the previous iterates.
\begin{corollary}\label{corollary:angle-monovariant}
If $p' \geq \frac{\log^2 ((n\sigma_1)/(c\sigma_r))}{n^2} \left(\frac{10r\mu'}{ c} \cdot \frac{\sigma_1}{\sigma_r}\right)^{10}$ then $\alpha(V_x^{t+1}, V_x) \leq 0.1\left(\alpha(V_y, V_y^t) + \alpha(V_z, V_z^t) \right)$ with at least  $1 - \frac{1}{10^5 \log ((n\sigma_1)/ (c\sigma_r))}$ probability.
\end{corollary}
\begin{proof}
Note $H = U_xB_t^T + E$.  

By Lemma \ref{lemma:frobeniusbound}, with probability at least $1 - \frac{1}{10^5 \log ((n\sigma_1)/ (c\sigma_r))}$
\[
\norm{E}_2 \leq \frac{10^3  \mu' r^3 \sigma_1 \sqrt{\log(n/ (c\sigma_r))}}{n\sqrt{p'}} \left(\alpha(V_y, V_y^t) + \alpha(V_z, V_z^t) \right)
\]
If this happens, the largest singular value of $E$ is at most
\[
\sigma \leq \norm{E}_2 \leq \frac{10^3  \mu' r^3 \sigma_1 \sqrt{\log(n/ (c\sigma_r))}}{n\sqrt{p'}} \left(\alpha(V_y, V_y^t) + \alpha(V_z, V_z^t) \right)
\]
Let $\rho_1 \geq \dots \geq \rho_r$ be the singular values of $U_xB_t^T$.  By Claim \ref{claim:nondegenerate2} 
\[
\rho_r \geq \left(1 - \alpha(V_y, V_y^t) - \alpha(V_z, V_z^t)\right)c^3\sigma_r
\]

Let $H_r$ be the rank-$r$ approximation of $H$ given by the top $r$ singular components.  $U_xB_t^T$ has rank $r$ and since $H_r$ is the best rank $r$ approximation of $H$ in Frobenius norm we have $H_r = H + E' = U_xB_t^T + E + E'$ where $||E'||_2 \leq ||E||_2$.  Now note
\[
\norm{H_r - U_xB_t^T}_{\op} \leq \norm{E + E'}_{2} \leq 2\norm{E}_2
\]
Thus by Claim \ref{claim:operator-anglebound} (applied to the matrices $H_r$, $U_xB_t^T$)
\[
\alpha(V_x, V_x^{t+1}) \leq \frac{2||E||_2}{\rho_r} \leq \frac{4 ||E||_2}{\sigma_r c^3} \leq 0.1\left(\alpha(V_y, V_y^t) + \alpha(V_z, V_z^t) \right)
\]

\end{proof}

We also upper bound the incoherence of $V_x^{t+1}$, relying on Lemma \ref{lemma:rowbound}.

\begin{corollary}\label{corollary:incoherence-monovariant}
If $p' \geq \frac{\log^2 ((n\sigma_1)/(c\sigma_r))}{n^2} \left(\frac{10r\mu'}{ c} \cdot \frac{\sigma_1}{\sigma_r}\right)^{10}$ then the subspace $V_x^{t+1}$ is $\left(\frac{2\mu r}{c^2}\cdot \frac{\sigma_1^2}{\sigma_r^2}\right)^{10}$-incoherent with at least $1 - \left( \frac{c\sigma_r}{n \sigma_1}\right)^{10}$ probability.
\end{corollary}
\begin{proof}
By Lemma \ref{lemma:rowbound}, with $1 - \left( \frac{c\sigma_r}{n\sigma_1}\right)^{10}$ probability, each row of $E$ has norm at most $\frac{0.1c^3\sigma_r}{\sqrt{n}}$.  This implies $\norm{E}_2 \leq 0.1c^3\sigma_r$.
\\\\
Let the top $r$ singular values of $H$ be $\rho_1', \dots , \rho_r'$ and let the top $r$ singular values of $U_xB_t^T$ be $\rho_1, \dots , \rho_r$. Note $\rho_r' \geq \rho_r - \norm{E}_{\op}$.  Also by Claim \ref{claim:nondegenerate2}
\[
\rho_r \geq \left(1 - \alpha(V_y, V_y^t) - \alpha(V_z, V_z^t)\right)c^3\sigma_r
\]
Thus
\[
\rho_r' \geq \rho_r - \norm{E}_{\op} \geq \left(1 - \alpha(V_y, V_y^t) - \alpha(V_z, V_z^t)\right)c^3\sigma_r - \norm{E}_2 \geq \frac{c^3\sigma_r}{2}
\]
Next observe that each row of $U_xB_t^T$ has norm at most $r\sigma_1\sqrt{\frac{\mu r}{n}}$ and since each row of $E$ has norm at most $\frac{0.1c^3\sigma_r}{\sqrt{n}}$, we deduce that each row of $H$ has norm at most $2r\sigma_1\sqrt{\frac{\mu r}{n}}$.  Now by Claim \ref{claim:row-incoherence}, $V_x^{t+1}$ is incoherent with incoherence parameter
\[
\frac{4r^3\mu\sigma_1^2}{\left(\frac{c^3 \sigma_r}{2} \right)^2} = \frac{16\mu r^3\sigma_1^2}{c^6\sigma_r^2}.
\]
Clearly 
\[
\frac{16\mu r^3\sigma_1^2}{c^6\sigma_r^2} \leq \left(\frac{2\mu r}{c^2} \cdot \frac{\sigma_1^2}{\sigma_r^2}\right)^{10}
\]
so we are done.
\end{proof}

We can now complete the proof of the main theorem of this section, Theorem \ref{thm:alternatingmin}.
\begin{proof}[Proof of Theorem \ref{thm:alternatingmin}]
Combining Corollary \ref{corollary:angle-monovariant} and Corollary \ref{corollary:incoherence-monovariant}, we immediately get the desired.
\end{proof}

Theorem \ref{thm:alternatingmin} immediately gives us the following corollary.
\begin{corollary}\label{coro:altminfinal}
Assume that $V_x^0, V_y^0, V_z^0$ satisfy
\begin{itemize}
    \item $\max\left(\alpha(V_x, V_x^{0}),\alpha(V_y, V_y^{0}), \alpha(V_z, V_z^{0}) \right) \leq 0.1$
    \item The subspaces $V_x^0, V_y^0, V_z^0$ are $\mu'$-incoherent where $\mu' = \left(\frac{2\mu r}{c^2} \cdot \frac{\sigma_1^2}{\sigma_r^2}\right)^{10}$
\end{itemize}
Then with probability $0.99$, when {\sc Kronecker Alternating Minimization} is run with parameters
\begin{align*}
    &p_2 = \left(\frac{\mu r \log n  }{c} \cdot \frac{\sigma_1}{\sigma_r} \right)^{300}\frac{1}{n^{3/2}} \\
    &p' = \frac{\log^2((n\sigma_1)/(c\sigma_r))}{n^2} \left(\frac{10r\mu'}{ c} \cdot  \frac{\sigma_1}{\sigma_r}\right)^{10}
\end{align*}
we have 
\[
\max\left(\alpha(V_x, V_x^k), \alpha(V_y, V_y^k), \alpha(V_z, V_z^k)\right) \leq \left(\frac{\sigma_r c}{10\sigma_1 n} \right)^{10^2}.
\]
\end{corollary}

\section{Experiments}\label{sec:experiments}

In this section, we describe our experimental results and in particular how our algorithm compares to existing algorithms. The code for the experiments can be found at \url{https://github.com/cliu568/tensor_completion}. 

First, it is important to notice that how well an algorithm performs can sometimes depend a lot on properties of the low-rank tensor. In our experiments, we find that there are many algorithms that succeed when the factors are nearly orthogonal but degrade substantially when the factors are correlated. To this end, we ran experiments in two different setups:

\begin{itemize}
    \item Uncorrelated tensors: generated by taking $T = \sum_{i=1}^4 x_i \otimes y_i \otimes z_i$ where $x_i,y_i,z_i$ are random unit vectors.
    \item  Correlated tensors: generated by taking $T = \sum_{i=1}^4 0.5^{i-1} x_i \otimes y_i \otimes z_i$ where $x_1,y_1,z_1$ are random unit vectors and for $i > 1$, $x_i,y_i,z_i$ are random unit vectors that have covariance $\sim 0.88$ with $x_1,y_1,z_1$ respectively.
\end{itemize} 

Unfortunately many algorithms in the literature either cannot be run on real data, such as ones based on solving large semidefinite programs \cite{barak2016noisy, potechin2017exact} and for others no existing code was available. Moreover some algorithms need to solve optimization problems with as many as $n^3$ constraints \cite{xia2019polynomial} and do not seem to be able to scale to the sizes of the problems we consider here. Instead, we primarily consider two algorithms: a variant of our algorithm, which we call {\sc Kronecker Completion}, and {\sc Standard Alternating Minimization}.  For {\sc Kronecker Completion}, we randomly initialize $\{\wh{x_i}\}, \{\wh{y_i}\}, \{\wh{z_i}\}$ and then run alternating minimization with updates given by Equation \ref{eq:subspacealtmin}. For given subspace estimates, we execute the projection step of {\sc Post-Processing via Convex Optimization} to estimate the true tensor (omitting the decomposition and convex optimization). It seems that neither the initialization nor the post-processing steps are needed in practice.  For {\sc Standard Alternating Minimization}, we randomly initialize $\{\wh{x_i}\}, \{\wh{y_i}\}, \{\wh{z_i}\}$ and then run alternating minimization with updates given by Equation \ref{eq:naivealtmin}.  To estimate the true tensor, we simply take $\sum_{i=1}^r \wh{x_i} \otimes \wh{y_i} \otimes \wh{z_i}$

For alternating minimization steps, we use a random subset consisting of half of the observations.    We call this subsampling.  Subsampling appears to improve the performance, particularly of {\sc Standard Alternating Minimization}.  We discuss this in more detail below.  

\begin{figure*}[h!]
\centering
    \textbf{Error over time for {\sc Kronecker Completion} and {\sc Standard Alternating Minimization}}
    \subfloat{%
        \includegraphics[scale = 0.4]{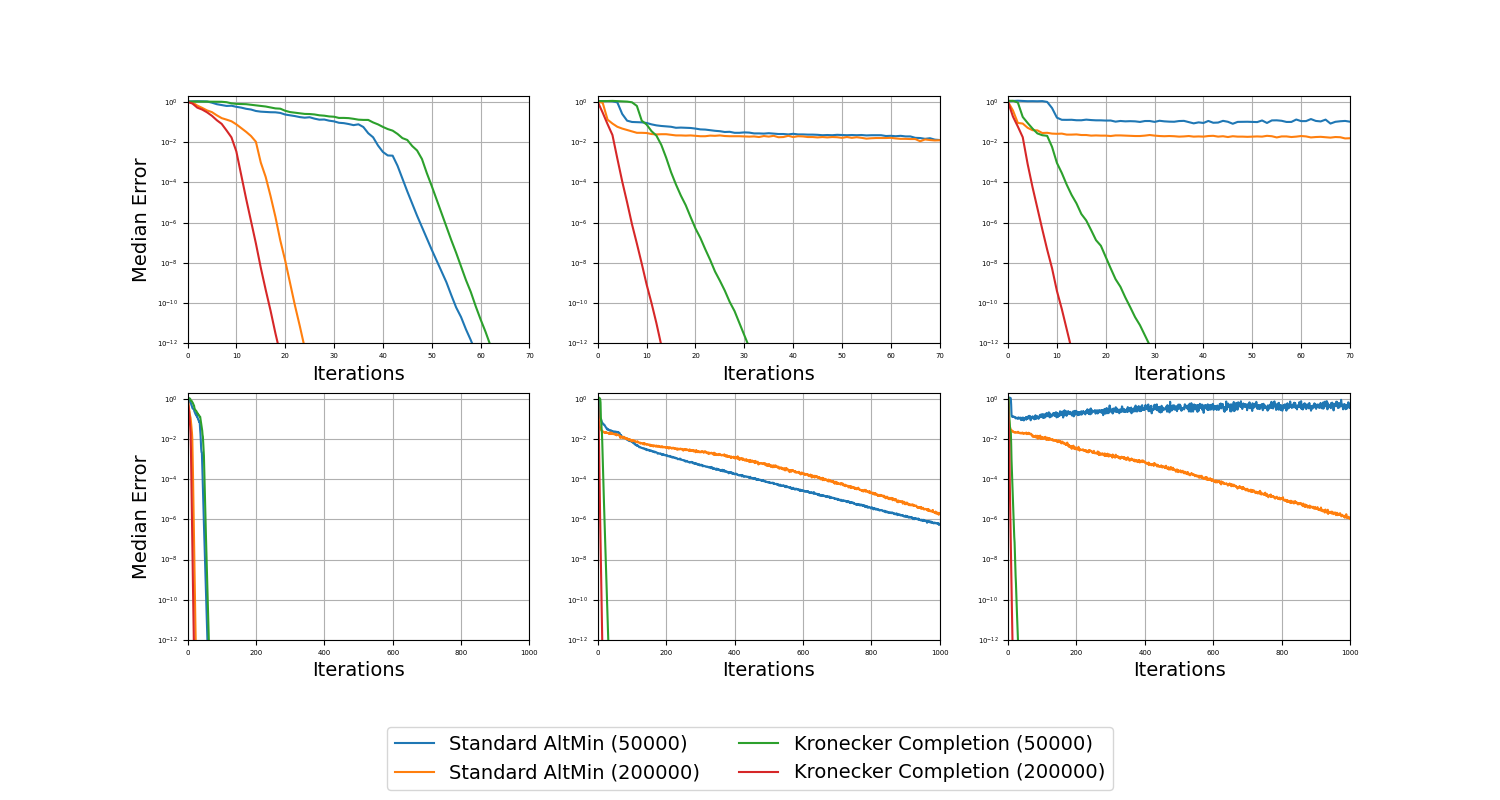}%
        }%
    \caption{Top Row: The plot on the left is for uncorrelated tensors.  The plot in the middle is for correlated tensors and with subsampling.  The plot on the right is for correlated tensors and no subsampling.  Bottom Row: The plots are the same as for the top row, but zoomed out so that it is easier to see the behavior of alternating minimization after a large number of iterations. }
\end{figure*}

We ran {\sc Kronecker Completion} and {\sc Standard Alternating Minimization} for $n=200,r=4$ and either $50000$ or $200000$ observations.  We ran $100$ trials and took the median normalized MSE i.e. $\frac{\norm{T_{\textsf{est}} - T}_2}{\norm{T}_2}$.  For both algorithms, the error converges to zero rapidly for uncorrelated tensors.  However, for correlated tensors, the error for {\sc Kronecker Completion} converges to zero at a substantially faster rate for {\sc Standard Alternating Minimization}.  Compared to {\sc Standard Alternating Minimization}, the runtime of each iteration of {\sc Kronecker Completion} is larger by a factor of roughly $r$ (here $r = 4$).  However, the error for {\sc Kronecker Completion} converges to zero in around $30$ iterations while the error for {\sc Standard Alternating Minimization} fails to converge to zero after $1000$ iterations, despite running for nearly $10$ times as long.  Naturally, we expect the convergence rate to be faster when we have more observations.  Our algorithm exhibits this behavior.  On the other hand, for {\sc Standard Alternating Minimization}, the convergence rate is actually slower with $200000$ observations than with $50000$.

In the rightmost plot, we run the same experiments with correlated tensors without subsampling.  Note the large oscillations and drastically different behavior of {\sc Standard Alternating Minimization} with $50000$ observations. 

\begin{figure*}[h!]
    \centering
    \textbf{{\sc Kronecker Completion} with noisy observations }
    \subfloat{%
        \includegraphics[scale = 0.5]{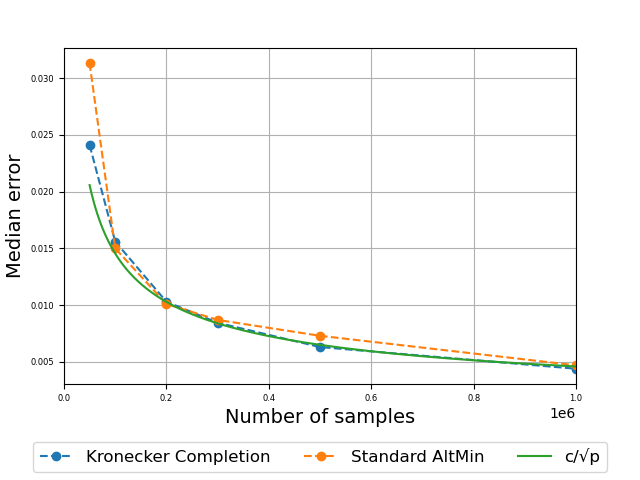}%
        }%
\caption{$n=200,r = 4$, we ran $100$ iterations of {\sc Kronecker Completion} and $400$ iterations of {\sc Standard Alternating Minimization}}
\end{figure*}
We also ran {\sc Kronecker Completion} with noisy observations ($10\%$ noise entrywise).  The error is measured with respect to the true tensor.  Note that the error achieved by both estimators is smaller than the noise that was added.  Furthermore, the error decays with the square root of the number of samples, which is essentially optimal (see \cite{cai2019nonconvex}).  

\begin{figure*}[h!]
    \centering
    \textbf{Error over time with noisy observations}
    \subfloat{%
        \includegraphics[scale = 0.4]{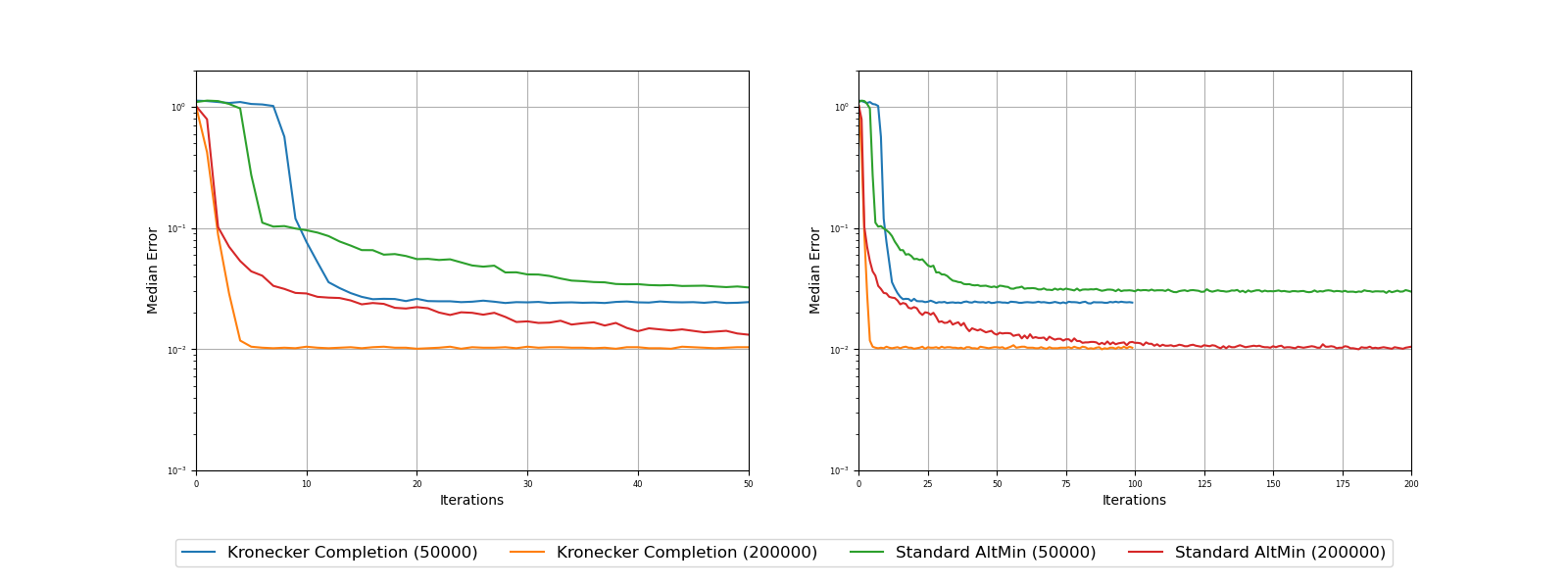}%
        }%
\caption{$n=200,r = 4$, the plot on the right is a zoomed out version of the plot on the left}
\end{figure*}

\noindent Although both algorithms converge to roughly the same error, similar to the non-noisy case, the error of our algorithm converges at a significantly faster rate in this setting as well. 

\newpage

\begin{figure*}[h!]
    \centering
    \textbf{Sample complexities of various algorithms}
    
    \includegraphics[scale = 0.6]{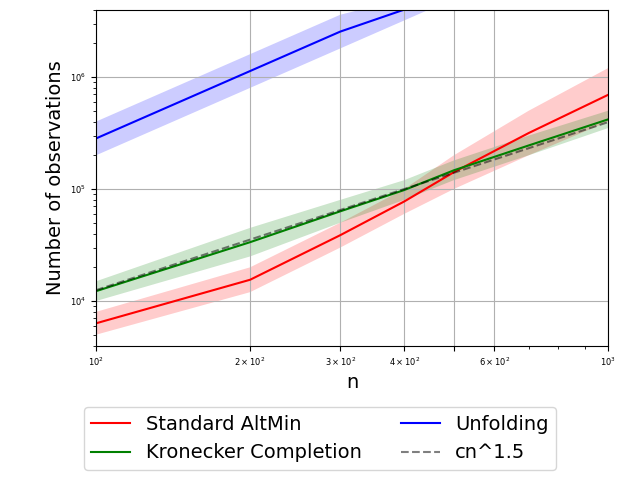}
    \caption{$r = 4$, success is defined as achieving $<1\%$ normalized MSE in at least half of the trials.  Shaded areas represent between $20\%$ and $80\%$ success rate.}
\end{figure*}

Finally, we ran various tensor completion algorithms on correlated tensors for varying values of $n$ and numbers of observations.  {\sc Unfolding} involves unfolding the tensor and running alternating minimization for matrix completion on the resulting $n \times n^2$ matrix.  We ran $100$ iterations for {\sc Kronecker Completion} and {\sc Unfolding} and $400$ iterations for {\sc Standard Alternating Minimization}.  Note how the sample complexity of {\sc Kronecker Completion} appears to grow as $n^{3/2}$ whereas even for fairly large error tolerance ($1\%$), {\sc Standard Alternating Minimization} seems to have more difficulty scaling to larger tensors.

It is important to keep in mind that in many settings of practical interest we expect the factors to be correlated, such as in factor analysis \cite{moitra2018algorithmic}. Thus our experiments show that existing algorithms only work well in seriously restricted settings, and even then they exhibit various sorts of pathological behavior such as getting stuck without using subsampling, converging slower when there are more observations, etc. It appears that this sort of behavior only arises when completing tensors and not matrices. In contrast, our algorithm works well across a range of tensors (sometimes dramatically better) and fixes these issues.

\bibliographystyle{plain}
\bibliography{bibliography}

\begin{thebibliography}{10}

\bibitem{allman2009identifiability}
Elizabeth~S Allman, Catherine Matias, John~A Rhodes, et~al.
\newblock Identifiability of parameters in latent structure models with many
  observed variables.
\newblock {\em The Annals of Statistics}, 37(6A):3099--3132, 2009.

\bibitem{anandkumar2014tensor}
Animashree Anandkumar, Rong Ge, Daniel Hsu, Sham~M Kakade, and Matus Telgarsky.
\newblock Tensor decompositions for learning latent variable models.
\newblock {\em Journal of Machine Learning Research}, 15:2773--2832, 2014.

\bibitem{barak2016noisy}
Boaz Barak and Ankur Moitra.
\newblock Noisy tensor completion via the sum-of-squares hierarchy.
\newblock In {\em Conference on Learning Theory}, pages 417--445, 2016.

\bibitem{bengua2017efficient}
Johann~A Bengua, Ho~N Phien, Hoang~Duong Tuan, and Minh~N Do.
\newblock Efficient tensor completion for color image and video recovery:
  Low-rank tensor train.
\newblock {\em IEEE Transactions on Image Processing}, 26(5):2466--2479, 2017.

\bibitem{bhaskara2014smoothed}
Aditya Bhaskara, Moses Charikar, Ankur Moitra, and Aravindan Vijayaraghavan.
\newblock Smoothed analysis of tensor decompositions.
\newblock In {\em Proceedings of the forty-sixth annual ACM symposium on Theory
  of computing}, pages 594--603, 2014.

\bibitem{cai2019nonconvex}
Changxiao Cai, Gen Li, H~Vincent Poor, and Yuxin Chen.
\newblock Nonconvex low-rank tensor completion from noisy data.
\newblock In {\em Advances in Neural Information Processing Systems}, pages
  1861--1872, 2019.

\bibitem{candes2009exact}
Emmanuel~J Cand{\`e}s and Benjamin Recht.
\newblock Exact matrix completion via convex optimization.
\newblock {\em Foundations of Computational mathematics}, 9(6):717, 2009.

\bibitem{chandrasekaran2012convex}
Venkat Chandrasekaran, Benjamin Recht, Pablo~A Parrilo, and Alan~S Willsky.
\newblock The convex geometry of linear inverse problems.
\newblock {\em Foundations of Computational mathematics}, 12(6):805--849, 2012.

\bibitem{daniely2013more}
Amit Daniely, Nati Linial, and Shai Shalev-Shwartz.
\newblock More data speeds up training time in learning halfspaces over sparse
  vectors.
\newblock In {\em Advances in Neural Information Processing Systems}, pages
  145--153, 2013.

\bibitem{fazel2002matrix}
Maryam Fazel.
\newblock Matrix rank minimization with applications [ph. d. thesis].
\newblock {\em Elec. Eng. Dept, Stanford University}, 2002.

\bibitem{gandy2011tensor}
Silvia Gandy, Benjamin Recht, and Isao Yamada.
\newblock Tensor completion and low-n-rank tensor recovery via convex
  optimization.
\newblock {\em Inverse Problems}, 27(2):025010, 2011.

\bibitem{grotschel2012geometric}
Martin Gr{\"o}tschel, L{\'a}szl{\'o} Lov{\'a}sz, and Alexander Schrijver.
\newblock {\em Geometric algorithms and combinatorial optimization}, volume~2.
\newblock Springer Science \& Business Media, 2012.

\bibitem{gurvits2003classical}
Leonid Gurvits.
\newblock Classical deterministic complexity of edmonds' problem and quantum
  entanglement.
\newblock In {\em Proceedings of the thirty-fifth annual ACM symposium on
  Theory of computing}, pages 10--19, 2003.

\bibitem{matrixaltmin}
Moritz Hardt.
\newblock Understanding alternating minimization for matrix completion.
\newblock In {\em 2014 IEEE 55th Annual Symposium on Foundations of Computer
  Science}, pages 651--660. IEEE, 2014.

\bibitem{hillar2013most}
Christopher~J Hillar and Lek-Heng Lim.
\newblock Most tensor problems are np-hard.
\newblock {\em Journal of the ACM (JACM)}, 60(6):1--39, 2013.

\bibitem{jain2013low}
Prateek Jain, Praneeth Netrapalli, and Sujay Sanghavi.
\newblock Low-rank matrix completion using alternating minimization.
\newblock In {\em Proceedings of the forty-fifth annual ACM symposium on Theory
  of computing}, pages 665--674, 2013.

\bibitem{jain2014provable}
Prateek Jain and Sewoong Oh.
\newblock Provable tensor factorization with missing data.
\newblock In {\em Advances in Neural Information Processing Systems}, pages
  1431--1439, 2014.

\bibitem{kreimer2013tensor}
Nadia Kreimer, Aaron Stanton, and Mauricio~D Sacchi.
\newblock Tensor completion based on nuclear norm minimization for 5d seismic
  data reconstruction.
\newblock {\em Geophysics}, 78(6):V273--V284, 2013.

\bibitem{li2017low}
Xutao Li, Yunming Ye, and Xiaofei Xu.
\newblock Low-rank tensor completion with total variation for visual data
  inpainting.
\newblock In {\em Thirty-First AAAI Conference on Artificial Intelligence},
  2017.

\bibitem{liu2012tensor}
Ji~Liu, Przemyslaw Musialski, Peter Wonka, and Jieping Ye.
\newblock Tensor completion for estimating missing values in visual data.
\newblock {\em IEEE transactions on pattern analysis and machine intelligence},
  35(1):208--220, 2012.

\bibitem{moitra2018algorithmic}
Ankur Moitra.
\newblock {\em Algorithmic aspects of machine learning}.
\newblock Cambridge University Press, 2018.

\bibitem{spectral}
Andrea Montanari and Nike Sun.
\newblock Spectral algorithms for tensor completion.
\newblock {\em Communications on Pure and Applied Mathematics},
  71(11):2381--2425, 2018.

\bibitem{convexopt}
Yurii Nesterov.
\newblock Introductory lectures on convex programming volume i: Basic course.
\newblock 1998.

\bibitem{ng2017adaptive}
Michael Kwok-Po Ng, Qiangqiang Yuan, Li~Yan, and Jing Sun.
\newblock An adaptive weighted tensor completion method for the recovery of
  remote sensing images with missing data.
\newblock {\em IEEE Transactions on Geoscience and Remote Sensing},
  55(6):3367--3381, 2017.

\bibitem{potechin2017exact}
Aaron Potechin and David Steurer.
\newblock Exact tensor completion with sum-of-squares.
\newblock {\em arXiv preprint arXiv:1702.06237}, 2017.

\bibitem{tan2016short}
Huachun Tan, Yuankai Wu, Bin Shen, Peter~J Jin, and Bin Ran.
\newblock Short-term traffic prediction based on dynamic tensor completion.
\newblock {\em IEEE Transactions on Intelligent Transportation Systems},
  17(8):2123--2133, 2016.

\bibitem{trickett2013interpolation}
Stewart Trickett, Lynn Burroughs, Andrew Milton, et~al.
\newblock Interpolation using hankel tensor completion.
\newblock In {\em 2013 SEG Annual Meeting}. Society of Exploration
  Geophysicists, 2013.

\bibitem{matrixconcentration}
Joel~A Tropp.
\newblock An introduction to matrix concentration inequalities.
\newblock {\em arXiv preprint arXiv:1501.01571}, 2015.

\bibitem{wang2015rubik}
Yichen Wang, Robert Chen, Joydeep Ghosh, Joshua~C Denny, Abel Kho, You Chen,
  Bradley~A Malin, and Jimeng Sun.
\newblock Rubik: Knowledge guided tensor factorization and completion for
  health data analytics.
\newblock In {\em Proceedings of the 21th ACM SIGKDD International Conference
  on Knowledge Discovery and Data Mining}, pages 1265--1274, 2015.

\bibitem{xia2019polynomial}
Dong Xia and Ming Yuan.
\newblock On polynomial time methods for exact low-rank tensor completion.
\newblock {\em Foundations of Computational Mathematics}, 19(6):1265--1313,
  2019.

\bibitem{xie2016accurate}
Kun Xie, Lele Wang, Xin Wang, Gaogang Xie, Jigang Wen, and Guangxing Zhang.
\newblock Accurate recovery of internet traffic data: A tensor completion
  approach.
\newblock In {\em IEEE INFOCOM 2016-The 35th Annual IEEE International
  Conference on Computer Communications}, pages 1--9. IEEE, 2016.

\bibitem{yuan2016tensor}
Ming Yuan and Cun-Hui Zhang.
\newblock On tensor completion via nuclear norm minimization.
\newblock {\em Foundations of Computational Mathematics}, 16(4):1031--1068,
  2016.

\bibitem{zhu2018fairness}
Ziwei Zhu, Xia Hu, and James Caverlee.
\newblock Fairness-aware tensor-based recommendation.
\newblock In {\em Proceedings of the 27th ACM International Conference on
  Information and Knowledge Management}, pages 1153--1162, 2018.

\end{thebibliography}

\newpage
\appendix
\begin{center}
\Large{\textbf{Appendix}}
\end{center}
\section{Concentration Inequalities}

\begin{claim}\label{claim:modifiedchernoff}
Say we have real numbers $-\gamma \leq x_1, \dots , x_n \leq \gamma$.  Consider the sum
\[
X = \epsilon_1 x_1 + \dots + \epsilon_n x_n
\]
where the $\epsilon_i$ are independent random variables that are equal to $1-p$ with probability $p$ and equal to $-p$ with probability $1 - p$.  Assume $p \geq \frac{1}{n}$.  Then for any $t \geq 1$,
\[
\Pr\left[|X| \geq \sqrt{pn}\gamma t \right] \leq 2e^{-t/4}
\]

\end{claim}
\begin{proof}
Let $\beta$ be a constant with $0 \leq \beta \leq \gamma^{-1}$.  Using the fact that $e^x \leq 1 + x + x^2$ between $-1$ and $1$, we have
\[
\E[e^{\beta x_i}] = pe^{(1-p)\beta x_i} + (1-p)e^{-p \beta x_i} \leq 1 + p(1-p)\beta^2 x_i^2 \leq e^{p\beta^2 x_i^2}
\]
Thus
\[
\E[e^{\beta X}] \leq e^{p\beta^2 (x_1^2 + \dots + x_n^2)}
\]
Similarly
\[
\E[e^{-\beta X}] \leq e^{p\beta^2 (x_1^2 + \dots + x_n^2)}
\]
Now set $ \beta = \frac{1}{2\sqrt{pn} \gamma}$.  Note that this is a valid assignment because we assumed $p \geq \frac{1}{n}$.  Thus
\[
\Pr\left[|X| \geq \sqrt{pn}\gamma t\right] \leq 2e^{p\beta^2 n\gamma^2 - \beta\sqrt{pn}\gamma t} \leq 2e^{-t/4}
\]

\end{proof}

\begin{claim}\label{claim:matrix-chernoff}[Matrix Chernoff (see \cite{matrixconcentration})]
Consider a finite sequence $\{X_k\}$ of independent self-adjoint $n \times n$ matrices.  Assume that for all $k$, $X_k$ is positive semidefinite and its largest eigenvalue is at most $R$.  Let $\mu_{\min} = \lambda_{\min}\left(\sum_{k}\E[X_k]\right)$ be the smallest eigenvalue of the expected sum.  Then
\[
\Pr \left[ \lambda_{\min}\left(\sum_{k}X_k\right) \leq (1 - \delta)\mu_{\min} \right] \leq ne^{-\frac{\delta^2 \mu_{\min}}{2R}}
\]
\end{claim}

\begin{claim}\label{claim:matrix-chernoff2}[Matrix Chernoff (see \cite{matrixconcentration})]
Consider a finite sequence $\{X_k\}$ of independent self-adjoint $n \times n$ matrices.  Assume that for all $k$, $X_k$ is positive semidefinite and its largest eigenvalue is at most $R$.  Let $\mu_{\max} = \lambda_{\max}\left(\sum_{k}\E[X_k]\right)$ be the largest eigenvalue of the expected sum.  Then
\[
\Pr \left[ \lambda_{\max}\left(\sum_{k}X_k\right) \leq (1 + \delta)\mu_{\max} \right] \leq ne^{-\frac{\delta^2 \mu_{\max}}{2R}}
\]
\end{claim}

\section{Initialization}\label{sec:initialization}
The main purpose of this section is to prove that the {\sc Initialization} algorithm obtains good initial estimates for the subspaces.  In particular we will prove Theorem \ref{thm:init}. First note that
\begin{itemize}
    \item $||U_x(T)||_{\op} \geq ||U_x(T) (y_1 \otimes z_1)||_2 = \norm{\sum_{i=1}^r \sigma_i((y_i \cdot y_1) (z_i \cdot z_1)) x_i}_2 \geq  \sigma_1 c$
    \item The largest norm of a row of $U_x(T)$ is at most $r\sigma_1\sqrt{\frac{\mu r}{n}}$
    \item The largest norm of a column of $U_x(T)$ is at most $\sigma_1\frac{\mu r^2}{n}$
    \item The largest entry of $U_x(T)$ is at most $r\sigma_1 \left(\frac{\mu r}{n}\right)^{3/2}$
\end{itemize}
Thus $U_x(T)$ is $(\lambda, 1, \rho)$-incoherent for $\lambda = \frac{\mu r^3}{c^2}, \rho = \frac{\mu^2 r^4}{c^2}$ according to Assumption 3 in \cite{spectral}.  Below we use $U_x$ to denote $U_x(T)$.

The key ingredient in the proof of Theorem \ref{thm:init} is the following result from \cite{spectral}, which says that the matrix $\wh{B}$ is a good approximation for $U_xU_x^T$.  
\begin{lemma}[Restatement of Corollary $3$ in \cite{spectral}]\label{lem:spectral} When the {\sc Initialization} algorithm is run with \[
p_1 = \left(\frac{2\mu r \log n}{c^2} \cdot \frac{\sigma_1^2}{\sigma_r^2} \right)^{10}\frac{1}{n^{3/2}}, 
\] we have
\[
||\wh{B} - U_xU_x^T||_{\op} \leq \frac{10^3\lambda \rho (\log n)^4}{p_1^2 n^3}||U_x||_{\op}^2
\]
with probability at least $1 - \left(\frac{1}{n} \right)^{25}$.

\end{lemma}

Let $D$ be the $r \times r$  diagonal matrix whose entries are the top $r$ signed eigenvalues of $\wh{B}$.  We first note that $\wh{B}$ can be approximated well by its top $r$ principal components.
\begin{claim}\label{claim:lowrank}
\[
||XDX^T - \wh{B}||_{\op} \leq ||\wh{B} - U_xU_x^T||_{\op}
\]
\end{claim}
\begin{proof}
Let $\sigma$ be the $r+1$\ts{st} largest singular value of $\wh{B}$.  Note $||XDX^T - \wh{B}||_{\op} = \sigma$.  Let $Y$ be the $r+1$-dimensional space spanned by the top $r+1$ singular vectors of $\wh{B}$.  Note that there is some unit vector $y \in Y$ such that $U_x^Ty = 0$ (since $U_x^T$ has rank $r$) so 
\[
||\wh{B} - U_xU_x^T||_{\op} \geq ||(\wh{B} - U_xU_x^T)y||_2 = ||\wh{B}y||_2 \geq \sigma
\]
\end{proof}

Let $\Pi$ be the $n \times n$ diagonal matrix with $0$ on diagonal entries corresponding to rows of $X$ with norm at least $\tau \sqrt{\frac{r}{n}}$ and $1$ on other diagonal entries.  Note $X_0 = \Pi X$.  

Lemma \ref{lem:spectral} and Claim \ref{claim:lowrank} imply that $XDX^T$ is a good approximation for $U_xU_x^T$.  To analyze the {\sc Initialization} algorithm, we will need to rewrite these bounds using $X_0$ in place of $X$.
\begin{claim}\label{claim:rowmod}
\[
||X_0  DX_0^T - U_xU_x^T||_{\op} \leq 0.1\left(\frac{c\sigma_r}{\sigma_1}\right)^{10}\sigma_1^2
\]
with probability at least $1 - \frac{1}{n^{25}}$
\end{claim}
\begin{proof}
First note that the squared Frobenius norm of $X$ is $r$ so there are at most $\frac{n}{\tau^2}$ entries of $\Pi$ that are $0$.  Next note 
\[
X_0  DX_0^T - U_xU_x^T =  ((\Pi X)D(\Pi X)^T   - (\Pi U_x)(\Pi U_x)^T) +  (\Pi U_x)(\Pi U_x)^T - U_xU_x^T
\]
Note that $(\Pi U_x)(\Pi U_x)^T - U_xU_x^T$ is $0$ in all but at most $\frac{2n^2}{\tau^2}$ entries.  Furthermore, all entries of $U_xU_x^T$ are at most $\frac{\mu^3\sigma_1^2r^5}{n}$.  Thus
\[
|| (\Pi U_x)(\Pi U_x)^T - U_xU_x^T||_2 \leq \frac{\mu^3\sigma_1^2r^5}{n} \cdot \frac{2n}{\tau} = \frac{2\mu^3\sigma_1^2r^5}{\tau}
\]
Also
\[
||(\Pi X)D(\Pi X)^T   - (\Pi U_x)(\Pi U_x)^T)||_{\op} \leq ||XDX^T - U_xU_x^T||_{\op} \leq 2||\wh{B} - U_xU_x^T||_{\op} \leq \frac{20^3\lambda \rho (\log n)^4}{p_0^2 n^3}||U_x||_{\op}^2
\]
where we used Claim \ref{claim:lowrank} and Lemma \ref{lem:spectral}.  Combining the previous two equations and noting that $||U_x||_{\op} \leq (\sigma_1r)$
\[
||X_0  DX_0^T - U_xU_x^T||_{\op} \leq  \frac{2\mu^3\sigma_1^2r^5}{\tau} + \frac{20^3\lambda \rho (\log n)^4}{p_0^2 n^3}(\sigma_1r)^2 \leq   0.1\left(\frac{c\sigma_r}{\sigma_1}\right)^{10}\sigma_1^2
\]
\end{proof}

To prove that $V_x^0$, the space spanned by the columns of $X_0$ is incoherent, we will need a bound on the smallest singular value of $X_0$.  In other words, we need to ensure that zeroing out the rows of $X$ whose norm is too large doesn't degenerate the column space. 
\begin{claim}\label{claim:perturbedsvd}
With probability $1 - \frac{1}{n^{25}}$, the $r$\ts{th} singular value of $X_0$ is at least $\frac{c^6}{10r^3}\left(\frac{\sigma_r}{\sigma_1}\right)^2$
\end{claim}
\begin{proof}
By Claim \ref{claim:nondegenerate}, the $r$\ts{th} singular value of $U_xU_x^T$ is at least $c^6\sigma_r^2$.  Thus by Claim \ref{claim:rowmod}, the $r$\ts{th} singular value of $X_0DX_0^T$ is at least $0.9c^6\sigma_r^2$. 
\\\\
Also note that the operator norm of $U_xU_x^T$ is at most $(\sigma_1 r)^2$ so by Lemma \ref{lem:spectral}, all entries of $D$ are at most $2(\sigma_1 r)^2$.  Also, the largest singular value of $X_0$ is clearly at most $r$.  Thus the smallest singular value of $X_0$ is at least 
\[
\frac{0.9c^6\sigma_r^2}{2(\sigma_1 r)^2 r} \geq \frac{c^6}{10r^3}\left(\frac{\sigma_r}{\sigma_1}\right)^2
\]

\end{proof}

Now we can use the results from Section \ref{sec:general-ineqs} to complete the proof of Theorem \ref{thm:init}.
\begin{corollary}\label{corollary:initialincoherence}
With probability at least $1 - \left(\frac{1}{n} \right)^{25}$, the subspace spanned by the columns of $X_0$ is $\left(\frac{2\mu r}{c^2} \cdot \frac{\sigma_1^2}{\sigma_r^2}\right)^{10}$-incoherent.
\end{corollary}
\begin{proof}
Note all rows of $X_0$ have norm at most $\tau \sqrt{\frac{r}{n}}$.  By Claim \ref{claim:row-incoherence} and Claim \ref{claim:perturbedsvd}, the column space of $X_0$ is incoherent with incoherence parameter
\[
\frac{\tau^2 r}{\left(\frac{c^6}{10r^3}\left(\frac{\sigma_r}{\sigma_1}\right)^2\right)^2} \leq \left(\frac{2\mu r}{c^2} \cdot \frac{\sigma_1^2}{\sigma_r^2}\right)^{10}.
\]
\end{proof}

\begin{corollary}\label{corollary:initialangle}
With probability at least $1 - \frac{1}{n^{25}}$ we have $\alpha(V_x, V_x^{0}) \leq 0.1$.
\end{corollary}
\begin{proof}
Note that $X_0DX_0^T$ and $U_xU_x^T$ are both $n \times n$ matrices with rank $r$.  By Claim \ref{claim:nondegenerate}, the $r$\ts{th} singular value of $U_x$ is at least $c^3\sigma_r$ and thus the $r$\ts{th} singular value of $U_xU_x^T$ is at least $c^6\sigma_r^2$.  Now using Claim \ref{claim:rowmod} and Claim \ref{claim:operator-anglebound}, we have with probability at least $1 - \frac{1}{n^{25}}$
\[
\alpha(V_x, V_x^0) \leq \frac{0.1\left(\frac{c\sigma_r}{\sigma_1}\right)^{10}\sigma_1^2}{c^6\sigma_r^2} \leq 0.1.
\]
\end{proof}

\begin{proof}[Proof of Theorem \ref{thm:init}]
Combining Corollary \ref{corollary:initialincoherence} and Corollary \ref{corollary:initialangle} we immediately get the desired.
\end{proof}

\section{Projection and Decomposition}\label{sec:projectionanddecomp}
After computing estimates for $V_x,V_y,V_z$, say $\widehat{V_x}, \widehat{V_y}, \widehat{V_z}$, we estimate the original tensor by projecting onto the space spanned by $\widehat{V_x} \otimes \widehat{V_y} \otimes \widehat{V_z}$.  Our first goal is to show that our error in estimating $T$ by projecting onto these estimated subspaces depends polynomially on the principal angles $\alpha(V_x, \wh{V_x}), \alpha(V_y, \wh{V_y}), \alpha(V_z, \wh{V_z})$.  This will then imply (due to Theorem \ref{thm:alternatingmin}) that we can estimate the entries of the original tensor to any inverse polynomial accuracy.

\subsection{Projection Step}\label{sec:projectionstep}

We will slightly abuse notation and use $\wh{V_x}, \wh{V_y}, \wh{V_z}$ to denote $n \times r$ matrices whose columns form orthonormal bases of the respective subspaces.  Let
\[
\delta = \max\left(\alpha(V_x, \wh{V_x}), \alpha(V_y, \wh{V_y}), \alpha(V_z, \wh{V_z})\right)
\]
Let $M = V_x \otimes V_y \otimes V_z$ be an $r^3 \times n^3$ matrix.  Let $S$ be a subset of the rows of $M$ where each row is chosen with probability $ p' = \frac{\log^2 n}{n^2} \left(\frac{10r\mu}{ c} \cdot \frac{\sigma_1}{\sigma_r}\right)^{10}$.  Let $M_S$ be the matrix obtained by taking only the rows of $M$ in $S$.
\\\\
The main lemma of this section is:
\begin{lemma}\label{lem:projectionstep}
Assume that $ \delta \leq \left(\frac{c\sigma_r}{n\sigma_1}\right)^{10}$.  When the {\sc Post-Processing via Convex Optimization} algorithm is run with $p_3  = 2\frac{\log^2 n}{n^2} \left(\frac{10r\mu}{ c} \cdot \frac{\sigma_1}{\sigma_r}\right)^{10}$, the tensor $T'$ satisfies
\[
||T - T'||_2 \leq 10\sigma_1\delta r n 
\]
with probability at least $1 - \frac{1}{n^{20}}$.
\end{lemma}

We first prove a preliminary claim.

\begin{claim}\label{claim:unfoldedcondition}
Assume that $ \delta \leq \left(\frac{c\sigma_r}{n\sigma_1}\right)^{10}$.  Then with probability at least $1 - \frac{1}{n^{20}}$, the smallest singular value of $M_S$ is at least $\frac{1}{n}$
\end{claim}
\begin{proof}
First we show that all of the entries of $\wh{V_x}$ are at most $2\sqrt{\frac{\mu r}{n}}$ in magnitude.  Assume for the sake of contradiction that some entry of  $\wh{V_x}$ is at least $2\sqrt{\frac{\mu r}{n}}$.  Let $c$ be the column of $\wh{V_x}$ that contains this entry.  Let $\Pi(c)$ be the projection of $c$ onto the subspace $V_x$.  Note that all entries of $\Pi(c)$ are at most $\sqrt{\frac{\mu r}{n}}$ since $V_x$ is $\mu$-incoherent.  This means that 
\[
||c - \Pi(c)||_2^2 \geq \frac{\mu r}{n}
\]
so
\[
\alpha(\wh{V_x}, V_x) \geq \sqrt{\frac{\mu r}{n}}
\]
contradicting the assumption that $\alpha(V_x, \wh{V_x}) \leq \left(\frac{c\sigma_r}{n\sigma_1}\right)^{10}$.  Similarly, we get the same bound for the entries of $\wh{V_y}, \wh{V_z}$.  This implies each row of $M$ has norm at most $\frac{8\mu^{3/2}r^{3}}{n^{3/2}}$.  Now consider $\frac{1}{p'}M_S^TM_S$.  This is a sum of independent rank-$1$ matrices with norm at most $\frac{1}{p'}\cdot \frac{64\mu^3r^6}{n^3}$ and the expected value of the sum is $I$, the identity matrix.  Thus by Claim \ref{claim:matrix-chernoff}
\[
\Pr\left[\lambda_{\min}\left( \frac{1}{p'}M_S^TM_S \right)\leq \frac{1}{2} \right] \leq \frac{1}{n^{20}}
\]
In particular, with probability at least, $1 - \frac{1}{n^{20}}$, the smallest singular value of $M_S$ is at least 
\[
\sqrt{\frac{p'}{2}} \geq \frac{1}{n}
\]
\end{proof}

\begin{proof}[Proof of Lemma \ref{lem:projectionstep}]
For $1 \leq i \leq r$, let $\ov{x_i}$ be the projection of $x_i$ onto $\wh{V_x}$.  Define $\ov{y_i},\ov{z_i}$ similarly.  Let
\[
\ov{T} = \sum_{i=1}^r \sigma_i \ov{x_i} \otimes \ov{y_i} \otimes \ov{z_i}
\]
Note
\begin{align*}
||\ov{x_i} \otimes \ov{y_i} \otimes \ov{z_i} - x_i \otimes y_i \otimes z_i||_2 &\leq ||\ov{x_i} - x_i||_2 +  ||\ov{y_i} - y_i||_2 +  ||\ov{z_i} - z_i||_2 \\ &\leq \alpha(V_x, \wh{V_x})+  \alpha(V_y, \wh{V_y})+ \alpha(V_z, \wh{V_z}) \\ &\leq 3\delta
\end{align*}
Thus
\begin{equation}\label{eq:comparison}
||\ov{T} - T||_2 \leq 3\delta\sigma_1 r
\end{equation}
Now consider the difference $\ov{T} - T'$.  By the definition of $T'$ We must have 
\[
\norm{\left(T' - T \right)|_S}_2 \leq \norm{\left(\ov{T} - T \right)|_S}_2 \leq  3\delta \sigma_1 r
\]
Thus $\norm{\left(T' - \ov{T} \right)|_S}_2 \leq 6\delta \sigma_1 r$.  Since $T', \ov{T}$ are both in the subspace $\wh{V_x} \otimes \wh{V_y} \otimes \wh{V_z}$, when flattened $T', \ov{T}$ can be written in the form $Mv', M\ov{v}$ respectively for some $v', \ov{v} \in \R^{r^3}$.  Now we know
\[
||M_S(v' - \ov{v})||_2 \leq 6\delta \sigma_1 r
\]
so by Claim \ref{claim:unfoldedcondition} we must have
\[
||v - \ov{v}||_2 \leq 6\delta \sigma_1 r n
\]
Thus
\[
||T' - \ov{T}||_2= ||M(v' - \ov{v})||_2 \leq 6\sigma_1\delta r n
\]
and combining with (\ref{eq:comparison}) we immediately get the desired.
\end{proof}

\subsection{Decomposition Step}
Now we analyze the decomposition step where we decompose $T'$ into rank-$1$ components.  First, we formally state {\sc Jennrich's Algorithm} and its guarantees.

\subsubsection{Tensor Decomposition via Jennrich's Algorithm}\label{sec:jennrichalg}
{\sc Jennrich's Algorithm} is an algorithm for decomposing a tensor, say $T= \sum_{i=1}^r (x_i \otimes y_i \otimes z_i)$, into its rank-$1$ components that works when the fibers of the rank $1$ components i.e. $x_1, \dots , x_r$ are linearly independent (and similar for $y_1, \dots , y_r$ and $z_1, \dots , z_r$).
\begin{algorithm}[H]
\caption{{\sc Jennrich's Algorithm} }
\begin{algorithmic} 
\State \textbf{Input}: Tensor $T' \in \R^{n \times n \times n}$ where 
\[
T' = T + E
\]
for some rank-$r$ tensor $T$ and error $E$ \\
\State Choose unit vectors $a,b \in \R^n$ uniformly at random
\State Let $T^{(a)}, T^{(b)}$ be $n \times n$ matrices defined as 
\begin{align*}
T^{(a)}_{ij} = T'_{i,j , \cdot } \cdot a \\
T^{(b)}_{ij} = T'_{i,j , \cdot } \cdot b
\end{align*}
\State Let $T_r^{(a)}, T_r^{(b)}$ be obtained by taking the top $r$ principal components of $T^{(a)}, T^{(b)}$ respectively.
\State Compute the eigendecompositions of $U = T_r^{(a)}(T_r^{(b)})^+$ and $V = \left((T_r^{(a)})^+T_r^{(b)}\right)^T$ (where for a matrix $M$, $M^+$ denotes the pseudoinverse)
\State Let $u_1, \dots , u_r, v_1, \dots , v_r$ be the eigenvectors computed in the previous step.  
\State Permute the $v_i$ so that for each pair $(u_i,v_i)$, the corresponding eigenvalues are (approximately) reciprocals.
\State Solve the following for the vectors $w_i$
\[
\arg\min \norm{T' - \sum_{i=1}^r u_i \otimes v_i \otimes w_i}_2^2
\]
\State Output the rank-$1$ components $\{u_i \otimes v_i \otimes w_i\}_{i=1}^r$
\end{algorithmic}
\end{algorithm}

Moitra \cite{moitra2018algorithmic} gives a complete analysis of {\sc Jennrich's Algorithm}.  The result that we need is that as the error $E$ goes to $0$ at an inverse-polynomial rate, {\sc Jennrich's Algorithm} recovers the individual rank-$1$ components to within any desired inverse-polynomial accuracy.  Note that the exact polynomial dependencies do not matter for our purposes as they only result in a constant factor change in the number of iterations of alternating minimization that we need to perform. 
\begin{theorem}[\cite{moitra2018algorithmic}]\label{thm:robust-jennrich}
Let 
\[
T = \sum_{i=1}^r \sigma_i (x_i \otimes y_i \otimes z_i)
\]
where the $x_i,y_i,z_i$ are unit vectors and $\sigma_1 \geq \dots \geq \sigma_r > 0$.  Assume that the smallest singular value of the matrix with columns given by $x_1, \dots, x_r$ is at least $c$ and similar for the $y_i$ and $z_i$.  Then for any constant $d$, there exists a polynomial $P$ such that if
\[
\norm{E}_2 \leq  \frac{\sigma_1}{P(n, \frac{1}{c}, \frac{\sigma_1}{\sigma_r})}
\]
then with $1 - \frac{1}{(10n)^d}$ probability, there is a permutation $\pi$ such that the outputs of {\sc Jennrich's Algorithm} satisfy
\[
\norm{\sigma_{\pi(i)} (x_{\pi(i)} \otimes y_{\pi(i)} \otimes z_{\pi(i)}) - u_i \otimes v_i \otimes w_i}_2 \leq \sigma_1\left(\frac{\sigma_r c}{10\sigma_1 n}\right)^d
\]
for all $1 \leq i \leq r$.
\end{theorem}
\begin{remark}
Note that the extra factors of $\sigma_1$ in the theorem above are simply to deal with the scaling of the tensor $T$.  
\end{remark}

\subsubsection{Uniqueness of Decomposition}
In Section \ref{sec:projectionstep}, we showed that our estimate $T'$ is close to $T$. We will now show that the components $\wh{\sigma_i},\wh{x}, \wh{y},\wh{z}$ that we obtain by decomposing $T'$ are close to the true components.

\begin{theorem}\label{thm:parameterestimates}
Consider running the {\sc Post-Processing via Convex Optimization} algorithm with parameter $$p_3 = 2\frac{\log^2 n}{n^2} \left(\frac{10r\mu}{ c} \cdot \frac{\sigma_1}{\sigma_r}\right)^{10}$$ and input subspaces that satisfy
\[
\max\left(\alpha(V_x, \wh{V_x}), \alpha(V_y, \wh{V_y}), \alpha(V_z, \wh{V_z})\right) \leq \left(\frac{\sigma_r c}{10\sigma_1 n} \right)^{10^2}.
\]
With probability at least $0.95$, there exists a permutation $\pi:[n] \rightarrow [n]$ and $\epsilon_x, \epsilon_y, \epsilon_z \in \{-1,1 \}$ such that the estimates  $\wh{\sigma_i},\wh{x_i}, \wh{y_i},\wh{z_i}$ computed in the {\sc Post-Processing via Convex Optimization} algorithm satisfy
\[
\frac{|\sigma_i - \widehat{\sigma_{\pi(i)}}|}{\sigma_1}, ||x_i - \epsilon_x \widehat{x_{\pi(i)}}||, ||y_i - \epsilon_y \widehat{y_{\pi(i)}}||, ||z_i - \epsilon_z \widehat{z_{\pi(i)}}|| \leq \epsilon
\]
where $\epsilon = \left(\frac{c\sigma_r}{10n\sigma_1}\right)^{20}$
\end{theorem}
\begin{proof}

By Lemma \ref{lem:projectionstep} we have with probability at least $0.98$ that
\[
||T - T'||_2 \leq \sigma_1 \left(\frac{\sigma_r c}{10\sigma_1 n} \right)^{99}
\]
Now by the robust analysis of {\sc Jennrich's Algorithm} (see Theorem \ref{thm:robust-jennrich}) we know that with $0.97$ probability, for the components $T_1, \dots , T_r$ that we obtain in the decomposition of $T'$, there is a permutation $\pi$ such that 
\[
||T_{\pi(i)} - \sigma_i x_i \otimes y_i \otimes z_i||_2 \leq \sigma_1 \left(\frac{\sigma_r c}{10n\sigma_1 } \right)^{40}
\]
for all $1 \leq i \leq r$.  We write $T_{\pi(i)} = \wh{\sigma_{\pi(i)}} \wh{x_{\pi(i)}} \otimes \wh{y_{\pi(i)}} \otimes \wh{z_{\pi(i)}}$ where $\wh{x_{\pi(i)}}, \wh{y_{\pi(i)}}, \wh{z_{\pi(i)}}$ are unit vectors and $\wh{\sigma_{\pi(i)}}$ is nonnegative.  Note that this decomposition is clearly unique up to flipping the signs on the unit vectors.  Then
\[
|| \wh{\sigma_{\pi(i)}} \wh{x_{\pi(i)}} \otimes \wh{y_{\pi(i)}} \otimes \wh{z_{\pi(i)}} - \sigma_i x_i \otimes y_i \otimes z_i||_2 \geq \bigg| \norm{\wh{\sigma_{\pi(i)}} \wh{x_{\pi(i)}} \otimes \wh{y_{\pi(i)}} \otimes \wh{z_{\pi(i)}}}_2 - \norm{\sigma_i x_i \otimes y_i \otimes z_i}_2 \bigg| = |\sigma_i - \wh{\sigma_{\pi(i)}}|
\]
Thus 
\[
|\sigma_i - \wh{\sigma_{\pi(i)}}| \leq \sigma_1 \left(\frac{\sigma_r c}{10n\sigma_1} \right)^{40}
\]
Now let $x_{\perp}$ be the projection of $\wh{x_{\pi(i)}}$ onto the orthogonal complement of $x_i$.  WLOG $\wh{x_{\pi(i)}} \cdot x_i \geq 0$.  Otherwise we can set $\epsilon_x = -1$.  Then
\[
x_{\perp} \geq 0.5||x_i - \wh{x_{\pi(i)}}||_2
\]
Also
\[
|| \wh{\sigma_{\pi(i)}} \wh{x_{\pi(i)}} \otimes \wh{y_{\pi(i)}} \otimes \wh{z_{\pi(i)}} - \sigma_i x_i \otimes y_i \otimes z_i||_2 \geq ||\wh{\sigma_{\pi(i)}} x_{\perp} \otimes \wh{y_{\pi(i)}} \otimes \wh{z_{\pi(i)}}||_2 \geq 0.1\sigma_r||x_i - \wh{x_{\pi(i)}}||_2
\]
Thus
\[
||x_i - \wh{x_{\pi(i)}}||_2 \leq \left(\frac{\sigma_r c}{10n\sigma_1} \right)^{20}
\]
and similar for $\wh{y_{\pi(i)}}, \wh{z_{\pi(i)}}$, completing the proof.
\end{proof}

\noindent From now on we will assume $\pi$ is the identity permutation and $\epsilon_x = \epsilon_y = \epsilon_z = 1$.  It is clear that these assumptions are without loss of generality.  We now have 
\begin{assertion}\label{assertion:parameters}
For all $1 \leq i \leq r$
\[
\frac{|\sigma_i - \widehat{\sigma_{i}}|}{\sigma_1}, ||x_i -  \widehat{x_{i}}||, ||y_i -  \widehat{y_{i}}||, ||z_i - \wh{z_i}|| \leq \epsilon
\]
where $\epsilon = \left(\frac{c\sigma_r}{10n\sigma_1}\right)^{20}$.
\end{assertion}

\section{Exact Completion via Convex Optimization}\label{sec:convexopt}
In the last step of our algorithm, once we have estimates $\widehat{\sigma_i}, \widehat{x_i}, \widehat{y_i}, \widehat{z_i}$, we solve the following optimization problem which we claim is strongly convex.  Let $S_{\sim}$ be the set of observed entries in $T_{\sim}$.  For each $1 \leq i \leq r$ let $y_i'$ be the unit vector in $\spn(\wh{y_1}, \dots , \wh{y_r})$ that is orthogonal to $\wh{y_1}, \dots , \wh{y_{i-1}}, \wh{y_{i+1}}, \dots , \wh{y_r}$.  Define $z_i'$ similarly.  We solve  
\begin{equation}\label{eq:optimization}
\min_{a_i,b_i,c_i}\norm{\left(T - \sum_{i=1}^r (\wh{\sigma_i}(\wh{x_i} + a_i)) \otimes (\wh{y_i} + b_i) \otimes (\wh{z_i} + c_i)\right)\Bigg|_{S_{\sim}}}_2^2
\end{equation}
with the constraints
\begin{itemize}
    \item $0 \leq ||a_i||_{\infty},||b_i||_{\infty}, ||c_i||_{\infty} \leq \left(\frac{c\sigma_r}{10n\sigma_1}\right)^{10}$ for all $1 \leq i \leq r$.
    \item $b_i \cdot y_i' = 0$ and $c_i \cdot z_i' = 0$ for all $1 \leq i \leq r$.
\end{itemize}

Assume that Assertion \ref{assertion:parameters} holds.  We will prove the following three lemmas which will imply that the optimization problem can be solved efficiently and yields the desired solution.

\begin{lemma}\label{lem:correctness}
Assuming that Assertion \ref{assertion:parameters} holds, an optimal solution of (\ref{eq:optimization}) is 
\begin{align*}
a_i &= \frac{\sigma_i}{\wh{\sigma_i}} \frac{(y_i \cdot y_i') (z_i \cdot z_i')}{(\wh{y_i} \cdot y_i')(\wh{z_i} \cdot z_i')}x_i -\wh{x_i} \\
b_i &= \frac{\wh{y_i} \cdot y_i'}{y_i \cdot y_i'} y_i - \wh{y_i} \\
c_i &= \frac{\wh{z_i} \cdot z_i'}{z_i \cdot z_i'} z_i - \wh{z_i}
\end{align*}
\end{lemma}
\begin{lemma}\label{lem:tractability}
Assuming that Assertion \ref{assertion:parameters} holds, with $1 - \frac{1}{n^{10}}$ probability over the random sample $S_{\sim}$, the objective function in (\ref{eq:optimization}) is $\frac{\sigma_r^2c^6}{10r} \cdot \frac{\log^2 n}{n^2} \left(\frac{10r\mu}{ c} \cdot \frac{\sigma_1}{\sigma_r}\right)^{10}$-strongly convex.
\end{lemma}

\begin{lemma}\label{lem:smoothness}
Assuming that Assertion \ref{assertion:parameters} holds, with $1 - \frac{1}{n^{10}}$ probability over the random sample $S_{\sim}$, the objective function in (\ref{eq:optimization}) is $20\sigma_1^2r \cdot \frac{\log^2 n}{n^2} \left(\frac{10r\mu}{ c} \cdot \frac{\sigma_1}{\sigma_r}\right)^{10}$-smooth.
\end{lemma}

\noindent First we demonstrate why these lemmas are enough to finish the proof of Theorem \ref{thm:main}.

\begin{proof}[Proof of Theorem \ref{thm:main}]
Note that for the solution stated in Lemma \ref{lem:correctness}, the value of the objective in (\ref{eq:optimization}) is $0$ and thus the solution is a local minimum.  Lemma \ref{lem:tractability} implies that the optimization problem is strongly convex and thus this is actually the global minimum.  Also since the ratio of the strong convexity and smoothness parameters is $\frac{200\sigma_1^2r^2}{\sigma_r^2c^6}$, the optimization can be solved efficiently (see \cite{convexopt}).  For the solution in Lemma \ref{lem:correctness}, the output of our {\sc Full Exact Tensor Completion Algorithm} is exactly $T$.  Thus combining Theorem \ref{thm:init}, Corollary \ref{coro:altminfinal}, and  Theorem \ref{thm:parameterestimates} with Lemma \ref{lem:correctness}, Lemma \ref{lem:tractability} and Lemma \ref{lem:smoothness}, we are done.
\end{proof}

\subsection{The True Solution Satisfies the Constraints}
First we show that the true solution $T$ can be recovered while satisfying the constraints.  
\begin{proof}[Proof of Lemma \ref{lem:correctness}]
When
\begin{align*}
a_i &= \frac{\sigma_i}{\wh{\sigma_i}} \frac{(y_i \cdot y_i') (z_i \cdot z_i')}{(\wh{y_i} \cdot y_i')(\wh{z_i} \cdot z_i')}x_i - \wh{x_i} \\
b_i &= \frac{\wh{y_i} \cdot y_i'}{y_i \cdot y_i'} y_i - \wh{y_i} \\
c_i &= \frac{\wh{z_i} \cdot z_i'}{z_i \cdot z_i'} z_i - \wh{z_i}
\end{align*}
then the value of the objective is $0$ and we exactly recover $T$.  It remains to show that this solution satisfies the constraints.  It is immediate that the second constraint is satisfied.  We now verify that the first constraint is also satisfied.  Note that the smallest singular value of $V_y$ is at least $c$.  Thus the smallest singular value of the matrix with columns $\wh{y_1}, \dots , \wh{y_r}$ is at least $c - \epsilon r$.  In particular $\wh{y_i} \cdot y_i'$ must be at least $c - \epsilon r$.  Also the difference between $\wh{y_i} \cdot y_i'$ and $y_i \cdot y_i'$ is at most $\epsilon$.  Thus
\[
 1 - \frac{2 \epsilon r}{c} \leq \frac{\wh{y_i} \cdot y_i'}{y_i \cdot y_i'} \leq  1 + \frac{2 \epsilon r}{c}
\]
Combining this with the fact that $\frac{|\sigma_i - \widehat{\sigma_i}|}{\sigma_1}, ||x_i - \widehat{x_i}||, ||y_i - \widehat{y_i}||, ||z_i - \widehat{z_i}|| \leq \epsilon$, it is clear that the first constraint is satisfied.
\end{proof}
\subsection{The Optimization Problem is Strongly Convex}\label{sec:strongconvexity}
To show that the optimization problem is strongly convex, we will compute the Hessian of the objective function.  Let $m$ be the magnitude of the largest entry of 
\[
T - \sum_{i=1}^r \wh{\sigma_i}\wh{x_i} \otimes \wh{y_i} \otimes \wh{z_i}
\]
Note $m \leq 10r\sigma_1\epsilon$ where $\epsilon = \left(\frac{c\sigma_r}{10n\sigma_1}\right)^{20}$. 
\\\\
Next, let $\wh{\sigma} = \max(\wh{\sigma_1}, \dots , \wh{\sigma_r})$.  Note $\wh{\sigma} \leq 2\sigma_1$.  Also define 
\[
D = \sum_{i = 1}^r  \wh{\sigma_i}a_i \otimes \wh{y_i} \otimes \wh{z_i} + \wh{\sigma_i}\wh{x_i} \otimes b_i \otimes \wh{z_i} + \wh{\sigma_i}\wh{x_i} \otimes \wh{y_i} \otimes c_i
\]
Note that the objective function can be written as 
\begin{equation}\label{eq:approx}
||D|_{S_{\sim}}||_2^2 + P(a_i,b_i,c_i)
\end{equation}
where $P$ is a polynomial with the following property: all terms of $P$ of degree $2$ have coefficients with magnitude at most $(10nr)^6m\wh{\sigma}$ and all coefficients for higher degree terms have magnitude at most $(10nr)^6\wh{\sigma}^2$.  Now to prove strong convexity we will lower bound the smallest singular value of the Hessian of  $||D|_{S_{\sim}}||_2^2$ (with respect to the variables $a_i,b_i,c_i$).  Since $m,a_i,b_i,c_i$ are all small, we can ensure that the contribution of $P$ does not affect the strong convexity and this will complete the proof.

\subsubsection{Understanding the Hessian when $S_{\sim}$ contains all entries}
First we consider $H_0$, the Hessian of $||D||_2^2$ i.e. when we are not restricted to the set of entries in $S_{\sim}$.
\begin{claim}\label{claim:fullhessianlowerbound}
The smallest eigenvalue of $H_0$ is at least $\frac{\sigma_r^2c^6}{2r}$ 
\end{claim}
\begin{proof}
Consider a directional vector $v = \left(\Delta_{a_1}, \Delta_{b_1}, \Delta_{c_1}, \dots , \Delta_{a_r}, \Delta_{b_r}, \Delta_{c_r} \right)$.  Then 
\[
||v^TH_0v||_2^2 = \norm{\sum_{i = 1}^r ( \wh{\sigma_i}\Delta_{a_i} \otimes \wh{y_i} \otimes \wh{z_i} + \wh{\sigma_i}\wh{x_i} \otimes \Delta_{b_i} \otimes \wh{z_i} + \wh{\sigma_i}\wh{x_i} \otimes \wh{y_i} \otimes \Delta_{c_i})  }_2^2
\]
We now lower bound the RHS.  For each $i$, let $\Delta_{a_i}^{-}$ be the projection of $\Delta_{a_i}$ onto $\spn(\wh{x_1}, \dots , \wh{x_r})$ and let $\Delta_{a_i}^{\perp}$ be the projection of $\Delta_{a_i}$ onto the orthogonal complement of $\spn(\wh{x_1}, \dots , \wh{x_r})$.  Define $\Delta_{b_i}^-, \Delta_{b_i}^{\perp}, \Delta_{c_i}^-, \Delta_{c_i}^{\perp}$ similarly.  We want to lower bound the squared Frobenius norm of
\begin{align*}
&\sum_{i = 1}^r (\wh{\sigma_i} \Delta_{a_i} \otimes \wh{y_i} \otimes \wh{z_i} + \wh{\sigma_i}\wh{x_i} \otimes \Delta_{b_i} \otimes \wh{z_i} + \wh{\sigma_i}\wh{x_i} \otimes \wh{y_i} \otimes \Delta_{c_i})  =\\ & \sum_{i = 1}^r (\wh{\sigma_i}\Delta_{a_i}^- \otimes \wh{y_i} \otimes \wh{z_i} + \wh{\sigma_i}\wh{x_i} \otimes \Delta_{b_i}^- \otimes \wh{z_i} + \wh{\sigma_i}\wh{x_i} \otimes \wh{y_i} \otimes \Delta_{c_i}^-) \\ &+ \sum_{i = 1}^r \wh{\sigma_i}\Delta_{a_i}^{\perp} \otimes \wh{y_i} \otimes \wh{z_i} + \sum_{i = 1}^r \wh{\sigma_i}\wh{x_i} \otimes \Delta_{b_i}^{\perp} \otimes \wh{z_i} + \sum_{i = 1}^r \wh{\sigma_i}\wh{x_i} \otimes \wh{y_i} \otimes \Delta_{c_i}^{\perp}
\end{align*}
Let the four sums above be $A,B,C,D$ respectively. $A,B,C,D$ are pairwise orthogonal.  Thus it suffices to lower bound the Frobenius norm of each of them individually.  First we lower bound the Frobenius norm of $A$.  Since $\Delta_{a_i}^-$ is in $\spn(\wh{x_1}, \dots , \wh{x_r})$, it can be written as a linear combination of $\wh{x_1}, \dots , \wh{x_r}$ say
\[
\Delta_{a_i}^- = a_i^{(1)}\wh{x_1} + \dots + a_i^{(r)}\wh{x_r}
\]
Furthermore 
\[
\left(a_i^{(1)}\right)^2+ \dots +  \left( a_i^{(r)}\right)^2 \geq \frac{||\Delta_{a_i}^-||_2^2}{r}
\]
We can use the same argument for $\Delta_{b_i}^-, \Delta_{c_i}^-$.  Also, since in our optimization problem we have the constraints $b_i \cdot y_i' = 0, c_i \cdot z_i' = 0$, we know that the coefficients $b_i^{(i)}, c_i^{(i)}$ are $0$.  Thus we can write $A$ as a sum
\[
A = \sum_{i=1}^r \sum_{j=1}^r\wh{\sigma_i} a_i^{(j)}\wh{x_j} \otimes \wh{y_i} \otimes \wh{z_i} + \sum_{i=1}^r \sum_{j=1, j \neq i}^r \wh{\sigma_i}\wh{x_i} \otimes b_i^{(j)}\wh{y_j} \otimes \wh{z_i} + \sum_{i=1}^r \sum_{j=1, j \neq i }^r \wh{\sigma_i}\wh{x_i} \otimes \wh{y_i} \otimes c_i^{(j)}\wh{z_j}
\] 
Note that the above is a linear combination of terms of the form $\wh{x_i} \otimes \wh{y_j} \otimes \wh{z_k}$ and each term appears at most once.  Furthermore, the sum of the squares of the coefficients is at least
\[
\sum_{i=1}^r\min(\wh{\sigma_1}, \dots , \wh{\sigma_r})^2 \left(  \frac{||\Delta_{a_i}^-||_2^2}{r} + \frac{||\Delta_{b_i}^-||_2^2}{r} + \frac{||\Delta_{c_i}^-||_2^2}{r} \right)
\]
Next, observe that the smallest singular value of the matrix with columns given by $\wh{x_i} \otimes \wh{y_j} \otimes \wh{z_k}$ for $1 \leq i,j,k \leq r$ is at least $(c - r\epsilon)^3$.  Thus
\[
||A||_2^2 \geq (c - r\epsilon)^6  \sum_{i=1}^r\min(\wh{\sigma_1}, \dots , \wh{\sigma_r})^2 \left(  \frac{||\Delta_{a_i}^-||_2^2}{r} + \frac{||\Delta_{b_i}^-||_2^2}{r} + \frac{||\Delta_{c_i}^-||_2^2}{r} \right)
\]
Now we lower bound the squared Frobenius norm of $B$.  Each slice of $B$ is a linear combination of $\wh{y_1} \otimes \wh{z_1}, \dots , \wh{y_r} \otimes \wh{z_r}$.  Note the matrix with columns given by $\wh{y_i} \otimes \wh{z_j}$ for $1 \leq i,j \leq r$ has smallest singular value at least $(c -r \epsilon)^2$.  Thus if we let $\Delta_{a_i}^{\perp [j]}$ be the $j$\ts{th} entry of $\Delta_{a_i}^{\perp}$ then the sum of the squares of the entries in the $j$\ts{th} layer of $B $ is at least
\[
\min(\wh{\sigma_1}, \dots , \wh{\sigma_r})^2(c - r\epsilon)^4\left( \left(\Delta_{a_1}^{\perp [j]}\right)^2 + \dots + \left(\Delta_{a_r}^{\perp [j]}\right)^2\right).
\]
Overall we get
\begin{align*}
||B||_2^2 \geq \sum_{j=1}^n \min(\wh{\sigma_1}, \dots , \wh{\sigma_r})^2(c - r\epsilon)^4\left( \left(\Delta_{a_1}^{\perp [j]}\right)^2 + \dots + \left(\Delta_{a_r}^{\perp [j]}\right)^2\right) \\ \geq  \min(\wh{\sigma_1}, \dots , \wh{\sigma_r})^2 (c - r\epsilon)^4\sum_{i=1}^r ||\Delta_{a_i}^{\perp}||_2^2.
\end{align*}
Similarly
\begin{align*}
||C||_2^2 &\geq  \min(\wh{\sigma_1}, \dots , \wh{\sigma_r})^2(c - r\epsilon)^4\sum_{i=1}^r ||\Delta_{b_i}^{\perp}||_2^2 \\ 
||D||_2^2 &\geq  \min(\wh{\sigma_1}, \dots , \wh{\sigma_r})^2(c - r\epsilon)^4\sum_{i=1}^r ||\Delta_{c_i}^{\perp}||_2^2.
\end{align*}
Overall we have
\[
||v^TH_0v||_2^2 = ||A||_2^2 + ||B||_2^2 + ||C||_2^2 + ||D||_2^2 \geq \frac{\min(\wh{\sigma_1}, \dots , \wh{\sigma_r})^2 (c - r\epsilon)^6}{r}||v||_2^2 \geq \frac{\sigma_r^2c^6}{2r}||v||_2^2
\]
and we get that the smallest eigenvalue of $H_0$ is at least $\frac{\sigma_r^2c^6}{2r}$.
\end{proof}
\subsubsection{Understanding the Hessian when $|S_{\sim}|$ is small}
To prove Lemma \ref{lem:tractability}, we want to go from a bound on the Hessian of $||D||_2^2$ to a bound on the Hessian of $||D|_{S_{\sim}}||_2^2$.  We will then use the fact that the Hessian of $P(a_i,b_i,c_i)$ is small and cannot substantially affect the strong convexity.  
\begin{proof}[Proof of Lemma \ref{lem:tractability}]
Note that $||D||_2^2$ is a sum of $n^3$ terms each of which is the square of a linear function (corresponding to an entry).  Each of these terms contributes a rank-$1$ term to the Hessian.  Furthermore, since all entries of $\wh{x_i}, \wh{y_i}, \wh{z_i}$ are at most $\sqrt{\frac{\mu r}{n}} + \epsilon$, the operator norm of each of these rank $1$ terms is at most $9\wh{\sigma}^2r^2\left(\sqrt{\frac{\mu r}{n}} + \epsilon\right)^4  \leq \frac{10r^4\mu^2 \sigma_1^2}{n^2}$.  
\\\\
If we add each entry to $S_{\sim}$ with probability $p = \frac{\log^2 n}{n^2} (\frac{10r\mu}{ c} \cdot \frac{\sigma_1}{\sigma_r})^{10}$ then by Claim \ref{claim:matrix-chernoff}, with at least $1 - \frac{1}{n^{10}}$ probability, the sum of the rank $1$ terms corresponding to entries of $S$ has smallest singular value at least $\frac{\sigma_r^2c^6}{4r} p$.
\\\\
We have shown that with high probability, the smallest eigenvalue of the Hessian of $||D|_{S_{\sim}}||_2^2$ is at least $\frac{\sigma_r^2c^6}{4r}p$.  It remains to note that the Hessian of $P(a_i,b_i,c_i)$ has operator norm at most $\frac{\sigma_r^2c^6}{10rn^2}$ for all $a_i,b_i,c_i$ in the feasible set and thus the optimization problem we formulated is strongly convex with parameter
\[
\frac{\sigma_r^2c^6}{10r} p = \frac{\sigma_r^2c^6}{10r} \cdot \frac{\log^2 n}{n^2} \left(\frac{10r\mu}{ c} \cdot \frac{\sigma_1}{\sigma_r}\right)^{10}
\]
\end{proof}

\subsection{The Optimization Problem is Smooth}
The proof that the objective function is smooth will follow a similar approach to that in Section \ref{sec:strongconvexity}.  We use the same notation as the previous section.  Again, the first step will be to consider the Hessian $H_0$ of $||D||_2^2$ when we are not restricted to the set of entries in $S_{\sim}$.

\begin{claim}\label{claim:fullhessianupperbound}
The largest eigenvalue of $H_0$ is at most $5\sigma_1^2r$.
\end{claim}
\begin{proof}
Consider a directional vector $v = \left(\Delta_{a_1}, \Delta_{b_1}, \Delta_{c_1}, \dots , \Delta_{a_r}, \Delta_{b_r}, \Delta_{c_r} \right)$.  Then 
\[
||v^TH_0v||_2^2 = \norm{\sum_{i = 1}^r ( \wh{\sigma_i}\Delta_{a_i} \otimes \wh{y_i} \otimes \wh{z_i} + \wh{\sigma_i}\wh{x_i} \otimes \Delta_{b_i} \otimes \wh{z_i} + \wh{\sigma_i}\wh{x_i} \otimes \wh{y_i} \otimes \Delta_{c_i})  }_2^2.
\]
Thus,
\begin{align*}
||v^TH_0v||_2^2 \leq \left(\sum_{i = 1}^r \norm{ \wh{\sigma_i}\Delta_{a_i} \otimes \wh{y_i} \otimes \wh{z_i}}_2 + \norm{\wh{\sigma_i}\wh{x_i} \otimes \Delta_{b_i} \otimes \wh{z_i}}_2 + \norm{\wh{\sigma_i}\wh{x_i} \otimes \wh{y_i} \otimes \Delta_{c_i}}_2  \right)^2 \\ \leq \left(\wh{\sigma}\sum_{i=1}^r \left(\norm{\Delta_{a_i}}_2 + \norm{\Delta_{b_i}}_2 + \norm{\Delta_{c_i}}_2\right)  \right)^2 \\ \leq 5\sigma_1^2r \left( \sum_{i=1}^r \norm{\Delta_{a_i}}_2^2 + \norm{\Delta_{b_i}}_2^2 + \norm{\Delta_{c_i}}_2^2 \right) \\ = 5\sigma_1^2r\norm{v}_2^2
\end{align*}
which immediately implies the desired.

\end{proof}

Now we can complete the proof of Lemma \ref{lem:smoothness} in the same way we proved Lemma \ref{lem:tractability} through a matrix Chernoff bound and the fact that the Hessian of $P(a_i,b_i,c_i)$ is small.

\begin{proof}[Proof of Lemma \ref{lem:smoothness}]
Note that $||D_{S_{\sim}}||_2^2$ is a sum of $|S_{\sim}|$ terms each of which is the square of a linear function (corresponding to an entry).  Each of these terms contributes a rank-$1$ term to the Hessian.  Furthermore, since all entries of $\wh{x_i}, \wh{y_i}, \wh{z_i}$ are at most $\sqrt{\frac{\mu r}{n}} + \epsilon$, the operator norm of each of these rank $1$ terms is at most $9\wh{\sigma}^2r^2\left(\sqrt{\frac{\mu r}{n}} + \epsilon\right)^4  \leq \frac{10r^4\mu^2 \sigma_1^2}{n^2}$. 
\\\\
If we add each entry to $S_{\sim}$ with probability $p = \frac{\log^2 n}{n^2} (\frac{10r\mu}{ c} \cdot \frac{\sigma_1}{\sigma_r})^{10}$ then by Claim \ref{claim:matrix-chernoff2}, with at least $1 - \frac{1}{n^{10}}$ probability, the sum of the rank $1$ terms corresponding to entries of $S$ has largest singular value at most $10\sigma_1^2rp$.
\\\\
We have shown that with high probability, the largest eigenvalue of the Hessian of $||D|_{S_{\sim}}||_2^2$ is at most $10\sigma_1^2rp$.  It remains to note that the Hessian of $P(a_i,b_i,c_i)$ has operator norm at most $\frac{\sigma_r^2c^6}{10rn^2}$ for all $a_i,b_i,c_i$ in the feasible set and thus the optimization problem we formulated is smooth with parameter
\[
20\sigma_1^2rp = 20\sigma_1^2r \cdot \frac{\log^2 n}{n^2} \left(\frac{10r\mu}{ c} \cdot \frac{\sigma_1}{\sigma_r}\right)^{10}
\]
\end{proof}

\section{Nearly Linear Time Implementation}\label{sec:lineartime}

Now we show how to implement our {\sc Full Exact Tensor Completion} algorithm with running time that is essentially linear in the number of observations (up to $\poly(r,\log n, \sigma_1/\sigma_r, \mu, 1/c)$- factors).  We will assume that our observations are in a list of tuples giving the coordinates and value i.e. $(i,j,k, T_{ijk})$.  Throughout this section, we use $\zeta$ to denote a quantity that is  $\poly(r,\log n, \sigma_1/\sigma_r, \mu, 1/c)$.

\subsection{Initialization}
We will construct $\wh{B}$ implicitly, i.e. we will store the coordinates of all of its nonzero entries and their values.  To do this we can enumerate over all pairs $(j,k) \in [n]^2$ such that there is some $i \in [n]$ for which we observe $T_{ijk}$.  For each of these pairs $(j,k)$ we take all pairs $i,i'$ such that $T_{ijk}$ and $T_{i'jk}$ are observed (we may have $i = i')$ and update $\wh{B}_{ii'}$ .  For each pair $(j,k)$, let $X_{j,k}$ be the number of distinct $i$ for which we observe $T_{ijk}$.  Note 
\[
\E\left[\sum_{j=1}^n \sum_{k=1}^n X_{j,k}^2 \right] \leq \zeta n^{3/2}
\]
so the time complexity of this step and the sparsity of $\wh{B}$ is essentially linear in the number of observations.
\\\\
Next to compute the top-$r$ singular vectors of $\wh{B}$ we can use matrix powering (with the implicit sparse representation for $\wh{B}$).  Note Lemma \ref{lem:spectral} implies that there is a sufficient gap between the $r$\ts{th} and $r+1$\ts{st} singular values of $\wh{B}$ that matrix powering converges within $\zeta$ rounds.  It is clear that the remainder of the steps in the initialization algorithm can be completed in nearly linear time.

\subsection{Alternating Minimization}

Note that for the least squares optimization problem, it suffices to solve the optimization for each row separately.  For the rows of $U_x(\wh{T_{t+1}})$, let $o_1, \dots , o_n$ be the number of observations in each row.  The least squares problems for the rows have sizes $o_1, \dots , o_n$ respectively.  Instead of constructing the full matrix $B_t$ (which has size $n^2$), we only need to compute the columns of $B_t$ that correspond to actual observations, which can be done using the matrices $V_x^t, V_y^t, V_z^t$.  Thus, the least squares problems can be solved in time essentially linear in $o_1 + \dots + o_n$.  Overall, this implies that all of the alternating minimization steps can be completed in nearly linear time.

\subsection{Post-Processing}
Note the projection step can be solved in nearly linear time and from it we obtain a representation of $T'$ as a sum of $r^3$ rank-$1$ tensors (corresponding to the basis given by $\wh{V_x} \otimes \wh{V_y} \otimes \wh{V_z}$.

\subsubsection{Jennrich's Algorithm}
Note we have an implicit representation of the tensor $T'$ that we are decomposing as a sum of $r^3$ rank-$1$ components.  Thus, we can compute implicit representations of $T^{(a)},T^{(b)}$ each as a sum of $r^3$ rank-$1$ matrices.  Next, we can use matrix powering with the implicit representations to compute the top $r$ principal components for $T^{(a)},T^{(b)}$ (note the analysis in \cite{moitra2018algorithmic} implies there is a sufficient gap between the $r$\ts{th} and $r+1$\ts{st} singular values of these matrices).  Now, we can compute the pseudo-inverses of the rank-$r$ matrices $T_r^{(a)},T_r^{(b)}$ (written implicitly as the sum of $r$ rank-$1$ matrices) in $n\poly(r)$ operations. 
\\\\
We can compute the eigendecompositions of $U = T_r^{(a)}(T_r^{(b)})^+$ and $V = \left((T_r^{(a)})^+T_r^{(b)}\right)^T$ using implicit matrix powering again (the analysis in \cite{moitra2018algorithmic} implies that with $0.99$ probability, the eigenvalues of these matrices are sufficiently separated).  These operations all take $n \zeta$ time.  Finally, we show that once we have $(u_1,v_1), \dots , (u_r,v_r)$, we can solve for $w_1, \dots , w_r$.  To do this, instead of solving the full least squares problem, we will choose a random subset of $\poly(r,\log n, \sigma_1/\sigma_r, \mu, 1/c)$ entries within each layer of the tensor $T'$ and solve the least squares optimization restricted to those entries.
\\\\
To see why this works, first note that the subspaces spanned by $u_1, \dots , u_r$ and $v_1, \dots , v_r$ are $2\mu$-incoherent (the proof of Theorem \ref{thm:parameterestimates} implies that $u_1, \dots , u_r$ and $v_1, \dots , v_r$ are close to the true factors up to some permutation).  Next, let $A$ be the matrix with columns given by $u_1 \otimes v_1, \dots , u_r \otimes v_r$.  Note that if $A'$ is a matrix constructed by selecting a random subset of $\poly(r,\log n, \sigma_1/\sigma_r, \mu, 1/c)$ rows of $A$, then with $1 - \frac{1}{n^{10}}$ probability, $A'$ is well-conditioned (by incoherence and the matrix Chernoff bound in Claim \ref{claim:matrix-chernoff}).  Since $A'$ is well-conditioned, the solution to the restricted least squares optimization problem must still be close to the true solution. 
\\\\
Thus, the entire least-squares optimization can be completed in $n\zeta$ time.  Overall, we conclude that the tensor decomposition step can be completed in $n\zeta$ time.

\subsubsection{Convex Optimization}
It remains to show that the final optimization problem can be solved in nearly linear time.  Note the size of the optimization problem is $n \poly(r,\log n, \sigma_1/\sigma_r, \mu, 1/c)$.  Lemma \ref{lem:tractability} and Lemma \ref{lem:smoothness} imply that the condition number of this convex optimization problem is $\poly(r,\log n, \sigma_1/\sigma_r, \mu, 1/c)$ so it can be solved in $n\zeta$ time.

\end{document}